\documentclass[11pt]{article}

\usepackage{amssymb,amsmath}
\usepackage{ifthen}
\usepackage{color}
\usepackage{xspace}
\usepackage{algorithm}
\usepackage[noend]{algorithmic}
\usepackage{graphicx}

\newcommand{\Eat}[1]{} %% just comment!

\newcommand{\Set}[1]{\ensuremath{\{ #1 \}}}
\newcommand{\SetCard}[1]{\ensuremath{| #1 |}}
\newcommand{\SpecSet}[2]{\ensuremath{\Set{#1 \mid #2}}}
\newcommand{\Ordinal}[1]{\ensuremath{{#1}^{\rm th}}}

\newcommand{\Half}[1][]{\ensuremath{%
\frac{\ifthenelse{\equal{#1}{}}{1}{#1}}{2}}\xspace}

\newcommand{\AAA}{\ensuremath{\mathcal{A}}\xspace}

\newcommand{\CCC}{\ensuremath{\mathcal{C}}\xspace}

\DeclareMathOperator*{\argmin}{argmin}
\DeclareMathOperator*{\argmax}{argmax}
\DeclareSymbolFont{AMSb}{U}{msb}{m}{n}
\DeclareMathSymbol{\N}{\mathord}{AMSb}{"4E}
\DeclareMathSymbol{\Z}{\mathord}{AMSb}{"5A}
\DeclareMathSymbol{\R}{\mathord}{AMSb}{"52}

\newcommand{\poly}{\ensuremath{\mbox{poly}}\xspace}

\newcommand{\Complexity}[1]{\ensuremath{\text{#1}}\xspace}
\newcommand{\Hardness}[2]{\ensuremath{\text{#1-#2}}\xspace}
\makeatletter
\newcommand*{\RNum}[1]{\expandafter\@slowromancap\romannumeral #1@}
\makeatother

\newcommand{\InsertAlgorithm}[3]%
{\begin{algorithm}[ht]
\caption{\sc #1}\label{#2}
\begin{algorithmic}[1]
\vspace{0.1cm}
\baselineskip=1.1\baselineskip
#3
\end{algorithmic}\end{algorithm}}
\def\AlgAssign{\ensuremath{\leftarrow}\xspace}

\newcommand{\InsertFigure}[3]{
\begin{figure}[h]
\includegraphics[width=#2]{./figures/#1.pdf}
\caption{#3}
\label{fig:#1}
\end{figure}}

\definecolor{USCcardinal}{rgb}{0.598,0.000,0.000}
\definecolor{USCgold}{rgb}{0.996,0.797,0.000}

\usepackage{./definitions}

\usepackage[letterpaper,top=1in,left=1in,right=1in,bottom=1in]{geometry}

\def\binarysearch{Binary Search\xspace}

\def\quasip{quasipolynomial\xspace}

\def\TOTQ{\ensuremath{K}\xspace}
\def\EachRound{\ensuremath{r}\xspace}
\def\NumRound{\ensuremath{k}\xspace}
\newcommand{\MAXDEGREE}{\ensuremath{\Delta}\xspace}
\def\QuasiDTIME{\ensuremath{O\Big(n^{O(\log n)}\Big)}\xspace}
\def\MAXCYCLE{\ensuremath{c}\xspace}
\def\ErrProb{\ensuremath{\delta}\xspace}
\def\MAXDIST{\ensuremath{D}\xspace}
\def\COSTFUNCTION{\ensuremath{\mathcal{C}}\xspace}
\def\Tolerance{\ensuremath{\lambda}\xspace}
\def\Capacity{\CCC}
\def\Marked{\ensuremath{M}\xspace}
\def\MultiWeights{%
\mbox{\ensuremath{\text{\sc Multiplicative-Weights}}}\xspace}
\def\TargetFinal{\ensuremath{\mu(t)}\xspace}
\def\SumFinal{\ensuremath{\Gamma}\xspace}
\def\NumDim{\ensuremath{D}\xspace}
\def\vc{\boldsymbol}
\def\true{\ensuremath{\text{true}}\xspace}
\def\false{\ensuremath{\text{false}}\xspace}
\def\firstplayer{vertex player\xspace}
\def\secondplayer{edge player\xspace}
\def\NumExists{\ensuremath{\hat{k}}\xspace}
\def\ExtraNode{\ensuremath{\hat{v}}\xspace}
\newcommand{\CritGadgetA}[1]{\ensuremath{T_{#1}}\xspace}
\newcommand{\CritGadgetB}[1]{\ensuremath{T'_{#1}}\xspace}
\newcommand{\CritNodeA}[1]{\ensuremath{t_{#1}}\xspace}
\newcommand{\CritNodeB}[1]{\ensuremath{t'_{#1}}\xspace}

\newcommand{\Weight}[1]{\ensuremath{%
\ifthenelse{\equal{#1}{}}{\omega}{\omega_{#1}}}\xspace}
\newcommand{\Reach}[3][]{\ensuremath{%
\ifthenelse{\equal{#1}{}}{N(#2,#3)}{N_{#1}(#2,#3)}}\xspace}
\newcommand{\ReachDist}[4][]{\ensuremath{%
\ifthenelse{\equal{#1}{}}{N(#2,#3,#4)}{N_{#1}(#2,#3,#4)}}\xspace}
\newcommand{\DirPotential}[2]{\ensuremath{%
\Psi_{#1}(#2)}\xspace}
\newcommand\NodeWeight[1][]{\ensuremath{%
\ifthenelse{\equal{#1}{}}{\mu}{\mu(#1)}}\xspace}
\newcommand{\Potential}[2]{%
\ensuremath{\Phi_{#1}(#2)}\xspace}
\newcommand{\WPotential}[2][]{\ensuremath{%
\ifthenelse{\equal{#1}{}}{\Phi(#2)}{\Phi_{#1}(#2)}}\xspace}
\newcommand{\SumWeight}[2][]{\ensuremath{%
\ifthenelse{\equal{#1}{}}{\Gamma(#2)}{\Gamma_{#1}(#2)}}\xspace}
\newcommand{\TSSA}[2]{%
\ensuremath{\text{\sc TargetSearch}}(#1, #2)\xspace}
\newcommand{\AtDistance}[2]{\ensuremath{\Gamma_{#1}(#2)}\xspace}
\newcommand{\NestCycle}[2]{%
\ensuremath{H_{#1, #2}}\xspace}
\newcommand{\GQBF}[1]{\ensuremath{\text{QBF}_{#1}}\xspace}
\newcommand{\Not}[1]{\ensuremath{\overline{#1}}\xspace}
\newcommand{\Clause}[1]{\ensuremath{C_{#1}}\xspace}
\newcommand{\ClGadgetA}[1]{\ensuremath{P_{#1}}\xspace}
\newcommand{\ClGadgetB}[1]{\ensuremath{P'_{#1}}\xspace}
\newcommand{\ClNodeA}[1]{\ensuremath{p_{#1}}\xspace}
\newcommand{\ClNodeB}[1]{\ensuremath{p'_{#1}}\xspace}

\newcommand{\mySubsection}[1]{\subsection*{#1}}

\begin{document}

\title{Deterministic and Probabilistic Binary Search in Graphs}

\author{%
Ehsan Emamjomeh-Zadeh%
\thanks{%
Department of Computer Science,
University of Southern California,
emamjome@usc.edu} 
\and
David Kempe%
\thanks{%
Department of Computer Science,
University of Southern California,
dkempe@usc.edu}
\and
Vikrant Singhal%
\thanks{
Department of Computer Science,
University of Southern California,
vikrants@usc.edu
}
}

\begin{titlepage}
\maketitle
\begin{abstract}
%% Abstract

We consider the following natural generalization of \binarysearch:
in a given undirected, positively weighted graph,
one vertex is a \emph{target}.
The algorithm's task is to identify the target
by adaptively querying vertices.
In response to querying a node $q$,
the algorithm learns either that $q$ is the target,
or is given an edge out of $q$ that
lies on a shortest path from $q$ to the target.
We study this problem in a general noisy model in which
each query independently receives a correct answer 
with probability $p > \Half$ (a known constant),
and an (adversarial) incorrect one with probability $1 - p$.

Our main positive result is that when $p = 1$ (i.e., all answers are
correct), $\log_2 n$ queries are always sufficient. 
For general $p$, we give an (almost information-theoretically optimal)
algorithm that uses, in expectation, no more than
$(1 - \ErrProb) \frac{\log_2 n}{1 - H(p)} + o(\log n) + O(\log^2 (1/\ErrProb))$
queries, and identifies the target correctly
with probability at least $1 - \ErrProb$.
Here, $H(p) = -(p \log p + (1 - p) \log(1 - p))$ denotes the entropy.
The first bound is achieved by the algorithm
that iteratively queries a 1-median of the nodes not ruled out yet;
the second bound by careful repeated invocations of 
a multiplicative weights algorithm.

Even for $p = 1$,
we show several hardness results for the problem of determining
whether a target can be found using \TOTQ queries.
Our upper bound of $\log_2 n$ implies a \quasip-time
algorithm for undirected connected graphs;
we show that this is best-possible under
the Strong Exponential Time Hypothesis (SETH).
Furthermore, for directed graphs, or for undirected graphs with
non-uniform node querying costs, the problem is
\Hardness{PSPACE}{complete}.
For a semi-adaptive version, in which one may query \EachRound nodes
each in \NumRound rounds, we show membership in
$\Sigma_{2\NumRound-1}$ in the polynomial hierarchy, and hardness for
$\Sigma_{2\NumRound-5}$.

\pagebreak

%% ** Only for STOC (camera-ready):
%\category{F.2}{Theory of Computation}{Analysis of Algorithms and Problem Complexity}
%\terms{Algorithms, Theory}
%\keywords{Noisy Binary Search,
%Searching in Metric Spaces,
%Exponential-Time Hypothesis,
%PSPACE-hardness,
%Quasipolynomial-Time Algorithms.}

\end{abstract}
\end{titlepage}

%%%%%%%%%%%%%%%%%%%%%%%%%%%%%%%%%%%%%%%%%%%%%%%%%%

%% Introduction

\section{Introduction}\label{sec:introduction}

One way of thinking about the classical \binarysearch algorithm
is as follows:
Given a path with a target vertex $t$ at an unknown position, 
an algorithm adaptively queries vertices, with the goal of discovering $t$.
In response to a query of $q$, the algorithm either learns that $q$ is
the target, or finds out whether the target is to the left or right of $q$.

This view of \binarysearch suggests a natural generalization to trees.
The algorithm knows the tree ahead of time,
and is trying to find the target $t$ in the tree.
When a vertex $q$ is queried which is not the target,
it reveals the subtree in which the target is located.
As shown by Jordan~\cite{jordan:1869:assemblages},
every tree has a separator vertex with the property that each of its
subtrees contains at most half of the vertices. Thus,
as pointed out by Onak and Parys~\cite{onak-parys:2006:trees-vertex-linear},
iteratively querying such a separator of the current
remaining subtree, then eliminating all subtrees known
not to contain the target
finds a target in a tree of $n$ vertices
within at most $\log n$ queries.\footnote{%
Throughout this paper, unless noted otherwise,
all logarithms are taken base 2.}
Since \binarysearch is optimal for the line,
this bound is also tight in the worst case.
For specific graphs, however, fewer queries may be enough:
for instance, only a single query suffices for a star.
The algorithmic problem of finding the optimal adaptive strategy
for a tree has been solved
by Onak and Parys~\cite{onak-parys:2006:trees-vertex-linear},
who present a linear-time algorithm.
This view of \binarysearch in turn suggests a further generalization
to arbitrary undirected graphs: 
searching for a target by adaptively querying vertices of
a connected undirected graph and receiving directional information.

A sequence of papers
\cite{feige-raghavan-peleg-upfal:1994:noisy,%
karp-kleinberg:2007:noisy,%
BenOr-hassidim:2008:noisy}
have studied a probabilistic version of
\binarysearch on a path, where each
answer to a query is correct only with probability $p > \Half$.
The upshot of this line of work is that a more careful variant of
\binarysearch uses $O(\log n)$ queries in expectation, with an
information-theoretically optimal constant depending on $p$. 
In our present work, we consider this natural generalization to the
noisy model for arbitrary graphs as well.

More formally, we study algorithms under the following general model:
given an arbitrary undirected
(for some of our results, directed) graph $G = (V, E)$
(with $n = \SetCard{V}, m = \SetCard{E}$)
with known positive edge weights\footnote{Thus, we can equivalently
think of searching for a target in a finite metric space, defined by
the shortest-path metric. An extension to infinite metric spaces is
discussed briefly in Section~\ref{sec:conclusion}.}
\Weight{e}, an unknown target node $t$ is to be located.
When a vertex $q$ is queried,
the algorithm receives a correct response with probability 
$p > \Half$, and an incorrect one with probability $1 - p$.
\begin{itemize}
\item If the query is answered correctly,
the response is that $q$ is the target, or an edge $e$,
incident on $q$, which lies on a \emph{shortest} $q$-$t$ path.
If there are multiple such edges $e$, then the one revealed to the
algorithm is arbitrary, i.e.,
it is a correct answer, but could be chosen adversarially
among correct answers.
\item If the query is answered incorrectly,
the response is completely arbitrary, and chosen by an adversary.
Of course, the adversary may choose to actually give a correct answer.
\end{itemize}

This model subsumes target search on a tree,
where there is a \emph{unique} path from $q$ to $t$.
Therefore, algorithms for it can be seen as natural
generalizations of the fundamental \binarysearch algorithm
to arbitrary graphs (and hence finite metric spaces).
As a first result, we show (in Section~\ref{sec:deterministic}) that
in the absence of incorrect answers,
the positive results from the line (and tree)
generalize to arbitrary weighted undirected graphs:

\begin{theorem}\label{thm:undirected-deterministic-informal}
For $p = 1$, in any connected undirected graph
with positive edge weights, a target can be found
in at most $\log n$ queries in the worst case.
\end{theorem}

In the noisy case as well, the positive results from the line
extend to arbitrary weighted undirected graphs,
albeit with a more complex algorithm and analysis. 
In Section~\ref{sec:noisy}, we prove the following theorem:

\begin{theorem}\label{thm:undirected-noisy-informal}
For any $p > \Half$
and any connected undirected graph with positive edge weights,
a target can be found with probability $1 - \ErrProb$
in $(1 - \ErrProb) \frac{\log n}{1 - H(p)} + o(\log n) + O(\log^2 (1/\ErrProb))$ 
queries in expectation.
($H(p) = -p \log p - (1-p)\log (1-p)$
is the entropy.)
\end{theorem}

The bound in Theorem~\ref{thm:undirected-deterministic-informal}
is tight even for the line.
Notice also that randomization cannot help:
it is easy to see that finding a target
selected uniformly at random from the line takes $\Omega(\log n)$
queries in expectation, and Yao's Minimax Principle therefore
implies the same lower bound for any randomized algorithm.
Moreover, as pointed out
by Ben-Or and Hasidim~\cite{BenOr-hassidim:2008:noisy},
the bound in Theorem~\ref{thm:undirected-noisy-informal}
is tight up to the $o(\log n) + O(\log^2 (1/\ErrProb))$ 
term even for the line.

\mySubsection{Hardness Results}

A natural optimization version of the problem is to ask,
given an instance, what is the minimum number of queries
necessary in the worst case
to find a target; or, as a decision problem, whether 
--- given a graph and target number \TOTQ of queries --- an
unknown target can be found within \TOTQ queries.
For trees, a sequence of papers discussed in detail
below~\cite{iyer-ratliff-vijayan:1988:node-ranking,%
schaffer:1989:node-ranking-linear,%
onak-parys:2006:trees-vertex-linear}
establishes that this problem can be solved in linear time.

For general graphs, even in the special case $p = 1$, 
the complexity of the problem
and some of its variants is quite intriguing.
The upper bound from
Theorem~\ref{thm:undirected-deterministic-informal} 
implies a simple $O(m^{\log n} \cdot n^2 \log n)$-time
algorithm using exhaustive search, thus making it unlikely that the
problem is \Hardness{NP}{hard}.
(The reason is that the exhaustive search tree's height
will be no more than $\log n$.)
However, we show in Section~\ref{sec:negative}
that the given \quasip time is essentially
the limit for the running time:
an $O(m^{o(\log n)})$-time algorithm would contradict the
Exponential-Time Hypothesis (ETH),
%~\cite{impagliazzo:paturi:ksat}
and
an $O(m^{(1 - \epsilon) \log n})$-time algorithm
for any constant $\epsilon > 0$ would contradict the
Strong Exponential-Time Hypothesis (SETH).
%~\cite{impagliazzo:paturi:ksat,calabro:impagliazzo:paturi}

A natural further generalization 
is to make the queries \emph{semi-adaptive}:
in each of \NumRound rounds, the algorithm gets to query
\EachRound nodes.
We consider this model with $p = 1$.
It is immediate that for $\NumRound = 1$, the problem is
NP-complete.
For constant $\NumRound > 2$, we show (in Section~\ref{sec:negative})
that the decision problem is in  $\Sigma_{2\NumRound - 1}$, the
\Ordinal{(2\NumRound - 1)} level of the polynomial hierarchy.
We give a nearly matching lower bound, by showing that it is also
$\Sigma_{2\NumRound - 5}$-hard.
Thus, semi-adaptive \binarysearch on undirected graphs
defines a natural class of problems
tracing out the polynomial hierarchy.

Even the fully adaptive version with $p = 1$ can easily become
computationally intractable with small modifications.
Specifically, we show in Section~\ref{sec:negative}
that the problem is PSPACE-complete for directed graphs\footnote{%
Note that for a directed cycle, $n - 1$ queries are necessary in
the worst case, as answers to a query
to a non-target node reveal nothing about the target.}
and also for undirected graphs in which nodes may have non-uniform
query costs, 
and the target is to be identified within a given budget.
All the hardness proofs are similar to each other,
and they are presented together.

\mySubsection{Edge Queries on Trees}

In Appendix~\ref{sec:edge-queries},
we turn our attention to a more frequently studied 
variant of the problem on trees,
in which queries are posed to edges instead of vertices. 
The algorithm can query an edge $e = (u, v)$, and
the response reveals whether the target is in the subtree
rooted at $u$ or at $v$. (We assume here that $p = 1$.)
This provides less information than querying one of the
two vertices, and the difference can be significant
when vertices have high degrees.

The edge query variant is considered in several recent papers,
which have focused on the algorithmic question of finding the optimal
adaptive strategy which minimizes the worst-case number of queries
\cite{BenAsher-farchi-newman:1999:trees-edge-poly,%
onak-parys:2006:trees-vertex-linear,%
mozes-onak-weimann:2008:trees-edge-linear}.
Upper and lower bounds on the maximum number of queries required
in trees of maximum degree \MAXDEGREE had been shown by 
Ben-Asher and Farchi to be 
$\log_{\MAXDEGREE/(\MAXDEGREE-1)} (n) \in \Theta(\MAXDEGREE \log n)$
and $\frac{\MAXDEGREE-1}{\log \MAXDEGREE} \log n$, 
respectively \cite{BenAsher-farchi:1997:trees}.
While the authors call these bounds ``matching,'' there is a gap of
$\Theta(\log \MAXDEGREE)$ between them.
In Appendix~\ref{sec:edge-queries}, we present an
improved algorithm for the problem, which uses at most 
$1 + \frac{\MAXDEGREE - 1}{\log (\MAXDEGREE + 1) - 1} \log n$
queries in the worst case, thus practically matching the
lower bound of \cite{BenAsher-farchi:1997:trees}.

\mySubsection{Related Work}

The algorithmic problem of identifying a target by \emph{edge queries}
in DAGs was initiated by Linial and Saks~\cite{linial-saks:1985:searching},
who studied it for several specific classes of graphs.
In general, the problem is known to be
\Hardness{NP}{hard}~\cite{carmo-donadelli-kohayakawa-laber:2004:poset%
,dereniowski:2008:edge-ranking}.
As a result, several papers have studied it specifically on trees. 

On trees, both the vertex and the edge query questions
(when $p = 1$) are equivalent 
to the respective vertex and edge ranking problems
(see, e.g.,
\cite{iyer-ratliff-vijayan:1988:node-ranking,%
schaffer:1989:node-ranking-linear,%
lam-yue:2001:edge-ranking-linear}). 
The vertex ranking problem requires that each vertex be assigned
a label $\ell_v$ such that if two vertices $u, v$
have the same label $\ell$,
then there is some vertex $w$ on the $u$-$v$ path with
strictly larger label.
The labels can be interpreted as the order in
which vertices will be queried (decreasingly). 
For the vertex ranking problem, 
Iyer, Ratliff, and Vijayan~\cite{iyer-ratliff-vijayan:1988:node-ranking}
gave an $O(n \log n)$ algorithm,
which was improved to linear time by Sch\"{a}ffer
\cite{schaffer:1989:node-ranking-linear}. 
Similarly, the edge ranking problem was solved for trees in linear
time by Lam and Yue~\cite{lam-yue:2001:edge-ranking-linear}.
Lam and Yue~\cite{lam-yue:1998:edge-ranking-hard}
also established that for general graphs, the problem is \Hardness{NP}{hard}.

Following the work of
Sch\"{a}ffer~\cite{schaffer:1989:node-ranking-linear} 
and Lam and Yue~\cite{lam-yue:2001:edge-ranking-linear},
linear-time algorithms for both versions (the vertex and edge query
models in trees) were rediscovered by
Onak and Parys~\cite{onak-parys:2006:trees-vertex-linear} and 
Mozes, Onak, and Weimann \cite{mozes-onak-weimann:2008:trees-edge-linear}.
The former paper gives an $O(n^3)$ algorithm for the edge query model 
(improving on the $O(n^4 \log^3 n)$ algorithm of Ben-Asher,
Farchi, and Newman~\cite{BenAsher-farchi-newman:1999:trees-edge-poly}),
while also providing a linear-time algorithm for the vertex query model.
The $O(n^3)$ algorithm was subsequently improved to linear time in
\cite{mozes-onak-weimann:2008:trees-edge-linear}.

Edge ranking and vertex ranking
elegantly encode querying strategies
for trees; however, they crucially exploit the fact that
the different possible responses to a vertex or edge query
identify disjoint components.
For general graphs, the resulting structures are not as simple.
Indeed, it is known that for general graphs,
finding an optimal edge ranking is
\Hardness{NP}{hard}~\cite{lam-yue:1998:edge-ranking-hard}.
More evidence of the hardness is provided by our
\Hardness{PSPACE}{completeness} result.

A natural generalization of the problem is
to introduce non-uniform costs on the vertices or edges. 
This problem is studied for lines and trees by Laber, Milidi\'{u}
and Pessoa~\cite{laber-milidiu-pessoa:2001:binary-search-cost} 
and Cicalese
et al.~\cite{cicalese-jacobs-laber-valentim:2012:tree-edge-cost},
respectively (among others).
In particular, Cicalese
et al.~\cite{cicalese-jacobs-laber-valentim:2012:tree-edge-cost}
show that with edge query costs,
the problem is \Hardness{NP}{hard} when the diameter of the tree
is at least $6$ and the maximum degree is at least $3$.
On the other hand,
\cite{cicalese-jacobs-laber-valentim:2012:tree-edge-cost}
presents a polynomial-time algorithm for the case
where the diameter of the given tree is at most $5$. 

An even more significant departure from the model we consider is
to assume that a probability distribution is given over target
locations, instead of the worst-case assumption made here.
The goal then typically is to minimize the expected number
(or total cost) of queries. 
This type of model is studied in several papers,
e.g., \cite{laber-milidiu-pessoa:2001:binary-search-cost,%
cicalese-jacobs-laber-molinaro:2011:tree-edge-average}. 
The techniques and results are significantly different from ours.

Noisy \binarysearch is a special case of a classical
question originally posed by R\'{e}nyi~\cite{renyi:1961:problem}, 
and subsequently restated by
Ulam~\cite{ulam:1991:adventures-mathematician}
as a game between two parties: one party asks questions,
and the other replies through a noisy channel
(thus lying occasionally).
Several models for noise are examined in detail
in \cite{pelc:2002:games-errors}. Rivest
et al.~\cite{rivest-meyer-kleitman-Winklmann-Spencer:1980:coping-error} study a model in which
the number of queries the adversary can answer
incorrectly is a constant (or a function of $n$,
but independent of the total number of queries).
When the adversary can answer at most a constant
fraction $r$ of queries incorrectly,
Dhagat et al.~\cite{dhagat-gacs-winkler:1992:twenty-question-lier},
proved the impossibility of finding the target for
$r \geq \frac{1}{3}$, and provided an $O(\log n)$ algorithm
when $r < 1/3$.
Pedrotti~\cite{pedrotti:1999:costant-rate-lies}
further improved the constant factors in this bound.
A modified version of this model, called the
\emph{Prefix-Bounded Model}, states that for every $i$,
at most an $r$ fraction of the initial $i$ query responses
may be incorrect.
In this model, Pelc~\cite{pelc:1989:error}
(using different terminology)
gives an $O(\log n)$ upper bound on the number of the queries
for every $r < \frac{1}{3}$.
Borgstrom and Kosaraju~\cite{borgstrom-kosaraju:1993:comparison-error}
generalize this bound to any $r < \Half$.
Aslam~\cite{aslam:1995:noisy-learning-searching}
also studies the Prefix-Bounded Model and
provides a reduction from the independent noise model to it.

The \NumRound-batch model, which is
similar to our notion of semi-adaptive querying,
is a different model in which the questioner can ask multiple
non-adaptive questions in a series of \NumRound batches.
Cicalese et al.~\cite{cicalese-mundici-vaccaro:2002:semi-adaptive}
consider the $2$-batch model,
under the assumption that at most a given number of query
responses are incorrect.

The independent noise model (used in our paper) is analyzed by
Feige et al.~\cite{feige-raghavan-peleg-upfal:1994:noisy}.
They use additional queries to backtrack in a search tree; repeated
queries to a vertex are used to obtain improved error probabilities
for queries.
They obtain a $\Theta(\log n)$ bound on the number of
queries with a large (non-optimal) constant depending on $p$.
Karp and Kleinberg~\cite{karp-kleinberg:2007:noisy}
consider the problem to find,
in an array of coins sorted by increasing probability of ``heads,''
the one maximizing the probability of ``heads'' subject to not
exceeding a given target. They show that this problem shares a lot
of similarities with noisy \binarysearch\footnote{%
In fact, this is a more general model
than noisy \binarysearch on a path.},
and use similar techniques
for positive results; they also provide information-theoretic lower bounds.
Their information theoretic lower bounds for the problem are matched
by Ben-Or and Hassidim \cite{BenOr-hassidim:2008:noisy},
whose noise model is the same  as ours. 
Ben-Or and Hassidim further generalize their algorithm to
achieve improved bounds on the query-complexity
of quantum \binarysearch.

Horstein~\cite{horstein:1963:noiseless-feedback}
provides another motivation for
noisy \binarysearch (on the continuous line):
a sender wants to share a (real) number with a receiver
over a noisy binary channel with a perfect feedback channel.
That is, the receiver receives each bit incorrectly
with probability $p$, but the sender knows which bit was received. 
Horstein~\cite{horstein:1963:noiseless-feedback}
analyzes noisy \binarysearch as an
algorithm under this model with a prior on the number to be transmitted:
the sender always sends the bit corresponding to the comparison
between the number and the median according to
the posterior distribution of the receiver. 
Jedynak et al.~\cite{jedynak-frazier-sznitman:2012:noise-entropy-loss}
show that this algorithm is optimal in the sense of
minimizing the entropy of the posterior distribution, while
Waeber et al.~\cite{waeber-frazier-henderson:2013:noisy-responses}
show that the expected value of the distance between this
algorithm's outcome and the target point decreases exponentially
in the number of the queries.

Noisy \binarysearch has been a significant area of study in
information theory and machine learning, where it has often been
rephrased as ``Active Bayesian Learning.''
Its uses can be traced back to early work on
parameter estimation in statistics, e.g., work by
Farrell~\cite{farrell:1964:sample-size} and
Burnashev and Zigangirov~\cite{burnashev-zigangirov:1974:estimation}.
It often takes the form of a search for the correct function
in a function space that maps a given instance space to an image space.
The initial algorithm for the noisy version with image space
$\Set{-1, +1}$ was proposed by
Nowak in \cite{nowak:2009:noisy-binary-search}.
Its running time is $O(\log\SetCard{H})$,
where $H$ is the hypothesis/function space.
A variation of the model was studied by
Yan et al.~\cite{yan-chaudhuri-javidi:2015:learning-noisy}.
Here, the instance space, image space and hypothesis space are,
respectively, $[0, 1]$, $\Set{0,1}$ and
$\SpecSet{f_\theta : [0,1] \rightarrow \Set{0,1}}{\theta \in [0,1],
f_\theta(x) = 1 \textrm{ if } x \geq \theta \textrm{, } 0 \textrm{ otherwise}}$.
They provide an $O(k^{-c})$ error bound,
where $k$ is the number of queries performed,
and $c$ is a constant determined by other parameters
defining the model. The goal is to approximate
the correct value of $\theta$, with the constraint that
the oracle might lie or abstain from providing an answer.
Both abstentions and lies might became more frequent as
the query points get closer to the actual threshold.
For general image spaces,
Naghshvar et al.
in \cite{naghshvar-javidi-chaudhuri:2015:active-learning-noise}
also obtain an $O(\log\SetCard{H})$ bound,
with improved constants for the special
case presented in \cite{nowak:2009:noisy-binary-search}.

\mySubsection{Notation and Preliminaries}

The algorithm is given a positively weighted connected undirected 
(or directed and strongly connected) graph $G = (V, E)$
with $\SetCard{V} = n$ and $\SetCard{E} = m$.
Edge weights are denoted by $\Weight{e} > 0$.
The distance $d(u, v)$ between vertices $u, v$ is the length of 
a shortest $u$-$v$ path with respect to the $\Weight{e}$.
For undirected graphs, $d(u, v) = d(v, u)$.

For a vertex $u$ and edge $e = (u, v)$ incident on $u$ 
(out of $u$ for directed graphs), we denote by
$\Reach{u}{e} = \SpecSet{w \in V}{d(u, w) = \Weight{e} + d(v, w)}$
the set of all vertices $w$ for which
$e$ lies on a shortest $u$-$w$ path.
Let $t$ be the (unknown) target vertex.
The guarantee made to the algorithm is that
when it queries a vertex $q$ 
which is not the target, with probability $p$, it will be
given an edge $e$ out of $q$ such that $t \in \Reach{q}{e}$.
Since $G$ is assumed (strongly) connected, such an answer
will always exist; we reiterate here that if
multiple edges $e$ exist such that $t \in \Reach{q}{e}$,
an arbitrary one (possibly chosen adversarially)
will be returned.
If $q$ is the target, then with probability $p$, this fact is 
revealed to the algorithm.
In both cases, with the remaining probability $1 - p$,
the algorithm is given an arbitrary (adversarial) answer,
which includes the possibility of the correct answer.

%%%%%%%%%%%%%%%%%%%%%%%%%%%%%%%%%%%%%%%%%%%%%%%%%%

%% Deterministic Model

\section{Deterministic Model}\label{sec:deterministic}

We first analyze the case $p = 1$, i.e.,
when responses to queries are noise-free.
We establish a tight logarithmic
upper bound on the number of queries required,
which was previously known
only for trees \cite{onak-parys:2006:trees-vertex-linear}.

\begin{theorem}\label{thm:undirected-weak}
There exists an adaptive querying strategy with the following property:
given any undirected,
connected and positively weighted graph $G$
with an unknown target vertex $t$,  
the strategy will find $t$ using at most $\log n$ queries.
(This bound is tight even when the graph is the line.)
\end{theorem}

\begin{proof}
In each iteration, 
based on the responses the algorithm has received so far,
there will be a set $S \subseteq V$
of \emph{candidate} nodes remaining
one of which is the target.
The strategy is to always query a vertex $q$ minimizing
the sum of distances to vertices in $S$, i.e., a $1$-median of $S$
(with ties broken arbitrarily).
Notice that $q$ itself may not lie in $S$.

More formally, for any set $S \subseteq V$ and vertex $u \in V$,
we define $\Potential{S}{u} = \sum_{v \in S} d(u, v)$.
The algorithm is given as Algorithm~\ref{alg:undirected-weak-upper}.

\InsertAlgorithm{Target Search for Undirected Graphs}{alg:undirected-weak-upper}{
\STATE{$S \AlgAssign V$.}
\WHILE{$\SetCard{S} > 1$}
\STATE{$q \AlgAssign \text{a vertex minimizing } \Potential{S}{q}$. \label{line:potential}}
\IF{$q$ is the target}
\RETURN{$q$}.
\ELSE
\STATE{$e = (q, v) \AlgAssign \text{response to the query}$.\label{line:response}} 
\STATE{$S \AlgAssign S \cap \Reach{q}{e}$.}\label{line:undirected-new-S}
\ENDIF
\ENDWHILE
\RETURN{the only vertex in $S$.}
}

First, note that \Potential{S}{q} and
\Reach{q}{e} can be computed using Dijkstra's algorithm,
and $q$ can be found by exhaustive search in linear time,
so the entire algorithm takes polynomial time.
We claim that it uses at most $\log n$ queries
to find the target $t$. To see this, consider an iteration
in which $t$ was not found;
let $e = (q, v)$ be the response to the query.
We write $S^{+} = S \cap \Reach{q}{e}$, and $S^{-} = S \setminus S^{+}$.
By definition, the edge $e$ lies on a shortest $q$-$u$ path for all 
$u \in S^{+}$, so that $d(v, u) = d(q, u) - \Weight{e}$
for all $u \in S^{+}$.
Furthermore, for all $u \in S^{-}$, the shortest path from $v$ to $u$
can be no longer than those going through $q$, so that
$d(v, u) \leq d(q, u) + \Weight{e}$ for all $u \in S^{-}$.
Thus, $\Potential{S}{v} \leq
\Potential{S}{q} - \Weight{e} \cdot \big(\SetCard{S^+} - \SetCard{S^-} \big)$.
By minimality of $\Potential{S}{q}$,
it follows that $\SetCard{S^{+}} \leq \SetCard{S^{-}}$, 
so $\SetCard{S^{+}} \leq \Half[\SetCard{S}]$. Consequently,
the algorithm takes at most $\log n$ queries.
\end{proof}

Notice that Theorem~\ref{thm:undirected-weak} implies a
\quasip time algorithm for finding an optimal adaptive strategy
for a given undirected graph.
Because an optimal strategy can use at most $\log n$ queries,
an exhaustive search over all possible vertices $q$ to query at each stage,
and all possible edges that could be given as responses
(of which there are at most as many as the degree of $q$)
will require at most $O(m^{\log n} n^2 \log n)$ time.
In particular, the decision problem of determining whether,
given an undirected graph $G$ and an integer \TOTQ,
there is an adaptive strategy using at most \TOTQ queries, 
cannot be \Hardness{NP}{hard} unless
$\Complexity{NP} \subseteq \QuasiDTIME$.
We later show hardness results for this problem
based on two hypotheses: ETH and SETH.

A natural question is how much the logarithmic bound
of Theorem~\ref{thm:undirected-weak} depends on our assumptions.
In Appendix~\ref{sec:almost-undirected}, we generalize this theorem
to strongly connected graphs which are \emph{almost undirected}
in the following sense: each edge $e$ with weight \Weight{e} belongs
to a cycle of total weight at most $\MAXCYCLE \cdot \Weight{e}$,
for a constant \MAXCYCLE.

%%%%%%%%%%%%%%%%%%%%%%%%%%%%%%%%%%%%%%%%%%%%%%%%%%

%% Noisy Model

\section{Noisy Model}\label{sec:noisy}

When $\Half < p < 1$, i.e., each response is incorrect
independently with probability $1 - p$, our goal is
to find the target with probability at least $1 - \ErrProb$,
for a given \ErrProb, minimizing the expected number of queries
in the worst case. (When $p \leq \Half$, even for the path,
it is impossible to get any useful information
about the target with any number of queries.%
\footnote{Recall that an ``incorrect'' query response
is adversarial; as the adversary may also respond correctly,
an algorithm cannot simply flip all answers.})

First, note that for the path (or more generally, trees),
a simple idea leads to an algorithm using
$c \log n \log \log n$ queries
(where $c$ is a constant depending on $p$):
simulate the deterministic algorithm,
but repeat each query $c \log \log n$ times and use the
majority outcome. A straightforward analysis
shows that with a sufficiently large $c$,
this finds the target with high probability.
For general graphs, however, this idea fails.
The reason is that if there are multiple
\emph{correct} answers for the query,
then there may be no majority among the answers,
and the algorithm cannot infer which answers are correct.

The $c \log n \log \log n$ bound for a path is not
information-theoretically optimal.
A sequence of papers
\cite{feige-raghavan-peleg-upfal:1994:noisy,%
karp-kleinberg:2007:noisy,%
BenOr-hassidim:2008:noisy}
improved the upper bound to $O(\log n)$ for simple paths.
More specifically,
Ben-Or and Hassidim~\cite{BenOr-hassidim:2008:noisy}
give an algorithm that
finds the target with probability $1 - \ErrProb$, using
$\frac{(1 - \ErrProb) \log n}{\Capacity(p)}
+ o(\log n) + O(\log(1/\ErrProb))$ queries,
where $\Capacity(p) = 1 - H(p) = 1 + p \log p + (1 - p)\log(1 - p)$.
As pointed out in \cite{BenOr-hassidim:2008:noisy},
this bound is information-theoretically optimal up to
the $o(\log n) + O(\log(1/\ErrProb))$ term.

In this section, we extend this result to
general connected undirected graphs.
At a high level, the idea of the algorithm is to keep track of
each node's \emph{likelihood} of being the target, based on
the responses to queries so far. 
Generalizing the idea of iteratively querying a median node from
Section~\ref{sec:deterministic}, we now iteratively query a
\emph{weighted} median, where the likelihoods are used as weights.
Intuitively, this should ensure that the total weight of
non-target nodes decreases exponentially faster
than the target's weight. In this sense, the algorithm
shares a lot of commonalities with the multiplicative weights update
algorithm, an idea that is also prominent in
\cite{karp-kleinberg:2007:noisy,BenOr-hassidim:2008:noisy}, for example.

The high-level idea runs into difficulty when there is one node that
accumulates at least half of the total weight. 
In this case, such a node will be the queried weighted median,
and the algorithm cannot ensure a sufficient decrease in the total weight.
On the other hand, such a high-likelihood node is
an excellent candidate for the target. Therefore,
our algorithm \emph{marks} such nodes
for ``subsequent closer inspection,''
and then removes them from the
multiplicative weights algorithm by setting their weight to $0$.

The key lemma shows that with high probability,
within $\Theta(\log n)$ rounds, the target node
has been marked. Therefore, the algorithm could now identify
the target by querying each
of the $O(\log n)$ marked nodes $\Theta(\log \log n)$ times,
and then keeping the node that was identified
as the target most frequently.

However, because there could be up to
$\Theta(\log n)$ marked nodes, this could still take
$\Theta(\log n \log \log n)$ rounds.
Instead, we run a
second phase of the multiplicative weights algorithm,
starting with weights only on the nodes that had been previously marked.
Akin to the first phase, we can show that with high probability,
the target is among the marked nodes in the first
$\Theta(\log \log n)$ rounds of the second phase.
Among the at most $\Theta(\log \log n)$ resulting candidates,
the target can now be identified with high probability by
querying each node sufficiently frequently. 
Because there are only $\Theta(\log \log n)$ nodes remaining,
this final round takes time only $O((\log \log n)^2)$.

Given a function
$\NodeWeight : V \rightarrow \R_{\geq 0}$,
for every vertex $u$, we define
$\WPotential[\NodeWeight]{u} =
\sum_{v \in V} \NodeWeight[v] d(u, v)$
as the weighted sum of distances
from $u$ to other vertices.
The weighted median is then
$\argmin_{u \in V} \WPotential[\NodeWeight]{u}$
(ties broken arbitrarily).
For any $S \subseteq V$, let 
$\SumWeight[\NodeWeight]{S} = \sum_{v \in S} \NodeWeight[v]$ 
be the sum of weights for the vertices in $S$.
The following lemma, proved by exactly the same arguments as
we used in Section~\ref{sec:deterministic},
captures the key property of the weighted median.

\begin{lemma} \label{lem:weighted-median-half}
Let $u$ be a weighted median
with respect to \NodeWeight
and $e = (u, v)$ an edge incident on $u$.
Then,
$\SumWeight[\NodeWeight]{\Reach{u}{e}} 
\leq \SumWeight[\NodeWeight]{V} / 2$.
\end{lemma}

We now specify the algorithm formally.
Algorithm~\ref{alg:mult-weights} encapsulates the
multiplicative weights procedure,
and Algorithm~\ref{alg:prob-binary} shows how to combine
the different pieces.
Compared to the high-level outline above, the target bounds of
$\Theta(\log n)$ and $\Theta(\log \log n)$ are modified
somewhat in Algorithm~\ref{alg:prob-binary}, mostly to account
for the case when $\ErrProb$ is very small.\Eat{
%% commented out:
In fact, as it will be clear momentarily,
this algorithm is designed only for $\ErrProb \leq 1/\log n$.
We will explain our trick for larger values of \ErrProb
in the proof of Theorem~\ref{thm:noisy}.}

%%%%%%%%%%%%%%%%%%%%%%%%%%%%%%%%
\Eat{ %% Note to ourselves:
We are using a Wikipedia version of Hoeffding bound
in this section which is
$\Pr[X \geq \mu + \Tolerance] \leq
e^{-2 \Tolerance^2 \TOTQ}$.
This implies that having \TOTQ coins, each of them head
with probability $p$, the probability that
less than $p - \Tolerance$ fraction of them are head is
at most \ErrProb if
$\TOTQ > \ln(1/\ErrProb) \frac{1}{2 \Tolerance^2}$.}
%%%%%%%%%%%%%%%%%%%%%%%%%%%%%%%%

\InsertAlgorithm{\MultiWeights ($S \subseteq V, \TOTQ$)}{alg:mult-weights}{
\STATE{$\NodeWeight[v] \AlgAssign 1/|S|$ for all vertices $v \in S$\\
and $\NodeWeight[v] \AlgAssign 0$ for all $v \in V \setminus S$.}
\STATE{$\Marked \AlgAssign \emptyset$.}
\FOR{$K$ iterations}
	\STATE{Let $q$ be a weighted median with respect to $\NodeWeight$.}
	\IF{$\NodeWeight[q] \geq \Half \cdot \SumWeight[\NodeWeight]{V}$}
		\STATE{Mark $q$, by setting
		$\Marked \AlgAssign \Marked \cup \Set{q}$.}
		\STATE{$\NodeWeight[q] \AlgAssign 0$.}
	\ELSE
		\STATE{Query node $q$.}
		\FORALL{nodes $v \in S$}
			\IF{$v$ is consistent with the response}
				\STATE{$\NodeWeight[v] \AlgAssign
				p \cdot \NodeWeight[v]$.}
			\ELSE
				\STATE{$\NodeWeight[v] \AlgAssign
				(1- p) \cdot \NodeWeight[v]$.}
			\ENDIF
		\ENDFOR
	\ENDIF
\ENDFOR
\RETURN{$\Marked$}
}

\InsertAlgorithm{Probabilistic Binary Search $(\ErrProb)$}%
{alg:prob-binary}{
\STATE{$\ErrProb' \AlgAssign \ErrProb/3$.}
\STATE{Fix $\Tolerance_1 =
\min(\sqrt{\frac{1}{\log \log n}}, \frac{\Capacity(p)}{2 \log(p/(1-p))})$.}
\label{line:lambda1}
\STATE{$\TOTQ_1 \AlgAssign
\max\Set{\frac{\log n}{\Capacity(p) - \Tolerance_1 \log(p/(1 - p))} + 1,
\frac{\ln(1 / \ErrProb')}{2 \Tolerance_1 ^ 2}}$.}
\STATE{$S_1 \AlgAssign \MultiWeights(V, \TOTQ_1)$.}
\STATE{Fix $\Tolerance_2 =
\frac{\Capacity(p)}{2 \log(p/(1-p))}$.}\label{line:lambda2}
\STATE{$\TOTQ_2 \AlgAssign \max\Set{
\frac{\log \SetCard{S_1}}{\Capacity(p) - \Tolerance_2 \log(p/(1 - p))} + 1,
\frac{\ln(1 / \ErrProb')}{2 \Tolerance_2^2}}$.}
\STATE{$S_2 \AlgAssign \MultiWeights(S_1, \TOTQ_2)$.}
\FORALL{$v \in S_2$}
	\STATE{Query $v$ repeatedly
	$\frac{2 \ln(\SetCard{S_2}/\ErrProb')}{(2p - 1)^2}$ times.}
	\IF{$v$ is returned as the target for at least half of these queries}
		\RETURN{$v$.}
	\ENDIF
\ENDFOR
\RETURN{failure.}
}

To analyze the performance of Algorithm~\ref{alg:mult-weights},
we keep track of the sum of the node weights
and show, by induction, that after $i$ iterations,
$\SumWeight[\NodeWeight]{V} \leq 2^{-i}$.
Consider the \Ordinal{i} iteration:
\begin{itemize}
\item If there exists a node whose weight is
at least half of the total node weight,
\MultiWeights sets its weight to $0$.
Therefore, in this case, the sum of node weights
drops to at most half of its previous value.
\item If \MultiWeights queries a median $q$,
and is told 
that $q$ is the target,
then $q$'s weight decreases by a factor of $p$,
and the other nodes' weights decrease by a factor of $1 - p$.
The new sum of node weights is
%\begin{align*}
$p \NodeWeight[q] +
(1 - p) (\SumWeight[\NodeWeight]{V} - \NodeWeight[q])
%& 
\leq \SumWeight[\NodeWeight]{V} / 2$,
%\end{align*}
because $\NodeWeight[q] \leq \SumWeight[\NodeWeight]{V} / 2$.
\item If \MultiWeights queries a median $q$,
and the response is an edge $e = (q, v)$,
the algorithm multiplies the weights of nodes
in $\Reach{q}{e}$ by $p$ 
and the weights of all other nodes by $1 - p$.
The new total weight is
%\begin{align*}
$p \SumWeight[\NodeWeight]{\Reach{q}{e}} +
(1 - p) \big(\SumWeight[\NodeWeight]{V} -
\SumWeight[\NodeWeight]{\Reach{q}{e}}\big)$,
%\end{align*}
which is at most $\SumWeight[\NodeWeight]{V} / 2$,
because Lemma~\ref{lem:weighted-median-half} implies that
$\SumWeight[\NodeWeight]{\Reach{q}{e}} \leq
\SumWeight[\NodeWeight]{V} / 2$.
\end{itemize}

The key lemma for the analysis
(Lemma~\ref{lem:multiplicative-weights})
shows that with high probability, the
target is in the set returned by the Multiplicative Weights algorithm.
Recall that $\Capacity(p)$ is defined
as $1 - H(p) = 1 + p \log p + (1 - p) \log (1 - p)$.

\begin{lemma} \label{lem:multiplicative-weights}
Let $S \subseteq V$ be a subset of vertices containing the target
and $0 < \Tolerance < \frac{\Capacity(p)}{\log(p / (1 - p))}$.
If
$\TOTQ \geq \frac{\ln(1/\ErrProb')}{2 \Tolerance^2}$
and
$\TOTQ > \frac{\log \SetCard{S}}{\Capacity(p) -
\Tolerance \log(p / (1 - p))}$,
then with probability at least $1 - \ErrProb'$, the target is
in the set returned by $\MultiWeights(S,\TOTQ)$.
\end{lemma}

\begin{proof}
By a standard Hoeffding Bound~\cite{hoeffding:1963:probability},
%\cite[Part 2 of Theorem 4.4]{mitzenmacher:upfal}, 
the probability that
fewer than $(p - \Tolerance) \TOTQ$ queries are answered correctly
is at most $e^{- 2 \TOTQ \Tolerance^2}$.
By the first bound on \TOTQ in Lemma~\ref{lem:multiplicative-weights},
this probability is at most $\ErrProb'$.

If we assume that the target is not marked by the end of
the \MultiWeights algorithm, its final weight is at least
$\TargetFinal \geq \frac{\big(p^{p - \Tolerance}(1 - p)^{1 - p + \Tolerance}\big)^\TOTQ}{\SetCard{S}}$,
with probability at least $1 - \ErrProb'$.
On the other hand, we showed that the final sum of node
weights is upper-bounded by $\SumFinal \leq 2^{-\TOTQ}$.
Using both bounds, we obtain that
\begin{align*}
\log(\TargetFinal/ \SumFinal) 
& \geq \TOTQ \cdot \log(2 p^{p - \Tolerance}(1 - p)^{1 - p + \Tolerance}) - \log \SetCard{S}\\
& = \TOTQ \cdot (1 - H(p) - \Tolerance \log (p / (1 - p))) - \log \SetCard{S}\\
& = \TOTQ \cdot (\Capacity(p) - \Tolerance \log (p / (1 - p))) - \log \SetCard{S}
\; > \; 0,
\end{align*}
where the final inequality follows from the second assumed bound on
\TOTQ in the lemma.
This implies that $\TargetFinal > \SumFinal$, a contradiction.
\end{proof}

The next lemma shows that the final phase of testing
each node finds the target with high probability.

\begin{lemma} \label{lem:repeatedly-querying}
Assuming that the target is in $S$,
if each node $v \in S$ is queried at least
$\TOTQ \geq \frac{2 \ln(\SetCard{S}/\ErrProb')}{(2p - 1)^2}$ 
times, then with probability at least $1 - \ErrProb'$, the true target
is the unique node returned as the ``target'' by at least half of
the queries.
\end{lemma}

\begin{proof}
We prove Lemma~\ref{lem:repeatedly-querying}
using standard tail bounds and a union bound.
Let $\Tolerance = p - \Half$.
As in the proof of Lemma~\ref{lem:multiplicative-weights},
for each queried node, a Hoeffding Bound establishes that
the probability that at most
$(p - \Tolerance) \TOTQ = \frac{\TOTQ}{2}$
queries are correctly answered is at most
$e^{-2 \TOTQ \Tolerance^2}$.
Because of the lower bound on \TOTQ,
this probability is upper-bounded by $\ErrProb' / \SetCard{S}$.
By a union bound, the probability that any queried vertex has
at least half of its queries answered incorrectly is at most
$\ErrProb'$. Barring this event, the target vertex is confirmed as
such more than half the time, while none of the non-target vertices
are confirmed as target vertices at least half the time.
\end{proof}

Combining these lemmas, we prove the following:

\begin{lemma} \label{lemma:noisy:small}
Given $\ErrProb > 0$, Algorithm~\ref{alg:prob-binary}
finds the target with probability at least $1 - \ErrProb$,
using no more than
$\frac{\log n}{\Capacity(p)} + o(\log n) + O(\log^2 (1/\ErrProb))$
queries.
\end{lemma}

\begin{proof}
Because $\Tolerance_1$ and $\TOTQ_1$, satisfy both bounds of
Lemma~\ref{lem:multiplicative-weights} by definition,
Lemma~\ref{lem:multiplicative-weights} can be applied,
implying that with probability at least $1-\ErrProb'$,
the set $S_1$ contains the target.
Note that the maximum in the definition of $\TOTQ_1$ is
taken over two terms. The first one is
$\frac{\log n}{\Capacity(p)} + o(\log n)$
because $\Tolerance_1 = o(1)$.
The second one is
$O(\log \log n \cdot \log (1/\ErrProb'))
= O((\log \log n)^2) + O(\log^2(1/\ErrProb'))$.
Therefore, $\SetCard{S_1}$ is bounded by
$\frac{\log n}{\Capacity(p)} + o(\log n) + O(\log^2 (1/\ErrProb'))$.

For the second invocation of $\MultiWeights$, the conditions of
Lemma~\ref{lem:multiplicative-weights} are again
satisfied by definition,
so $S_2$ will contain the target with probability at least
$1-2\ErrProb'$.
The maximum in the definition of $\TOTQ_2$ is
again taken over two terms.
The first one is $O(\log \log n + \log \log (1/\ErrProb'))$
(by the bound on $\SetCard{S_1}$)
and the second one is $O(\log(1/\ErrProb'))$.

Finally, by Lemma~\ref{lem:repeatedly-querying}, the target will be
returned with probability at least $1-3\ErrProb' = 1-\ErrProb$.
The final phase again only makes $o(\log n) + O(\log^2 (1/\ErrProb))$ queries,
giving us the claimed bound on the total number of queries.
\end{proof}

To obtain the bound
$(1- \ErrProb)
\frac{\log n}{\Capacity(p)} + o(\log n) + O(\log^2(1/\ErrProb))$
from Theorem~\ref{thm:undirected-noisy-informal},
we can reason as follows:
\begin{itemize}
\item If $\ErrProb < \frac{1}{\log n}$, then
$\frac{\log n}{\Capacity(p)} 
= \frac{(1-\ErrProb) \log n}{\Capacity(p)} +
  \frac{\ErrProb \log n}{\Capacity(p)}$,
and $\frac{\ErrProb \log n}{\Capacity(p)} = O(1)$ can be absorbed into
the $o(\log n)$ term.
%\begin{align*}
%\frac{\log n}{\Capacity(p)} + o(\log n) + O(\log^2(1/\ErrProb))
%= \\
%\frac{(1-\ErrProb) \log n}{\Capacity(p)} +
%\frac{\ErrProb \log n}{\Capacity(p)}
%+ o(\log n) + O(\log^2(1/\ErrProb)) = \\
% \frac{(1-\ErrProb) \log n}{\Capacity(p)} + o(\log n) + O(\log^2(1/\ErrProb)).
%\end{align*}
\item If $\ErrProb \geq 1/\log n$, we can modify the algorithm
using a simple trick from \cite{BenOr-hassidim:2008:noisy}:
With probability $\ErrProb - \frac{1}{\log n}$, the modified algorithm
outputs an arbitrary node without any queries; 
otherwise, it runs Algorithm~\ref{alg:prob-binary}
with error parameter $\hat{\ErrProb} = \frac{1}{\log n}$.
The success probability is then
$(1 - \ErrProb + \frac{1}{\log n})(1 - \frac{1}{\log n}) \geq 1 - \ErrProb$,
while the number of queries in expectation is
\begin{align*}
(1 - \ErrProb + \frac{1}{\log n})
\cdot \big( \frac{\log n}{\Capacity(p)} + o(\log n) + O(\log^2(1/\ErrProb))\big)
= \\
(1 - \ErrProb) \cdot \frac{\log n}{\Capacity(p)} + o(\log n) + O(\log^2(1/\ErrProb)).
\end{align*}
\end{itemize}

Hence, we have proved our main theorem:

\begin{theorem} \label{thm:noisy}
There exists an algorithm with the following property:
Given a connected undirected graph and positive $\ErrProb > 0$,
the algorithm finds the target with probability at least 
$1 - \ErrProb$, using no more than
$(1 - \ErrProb) \frac{\log n}{\Capacity(p)} +
o(\log n) + O(\log^2(1/\ErrProb))$
queries in expectation.
\end{theorem}

Notice that except for having $O(\log^2(1/\ErrProb))$ instead
of $O(\log(1/\ErrProb))$, this bound matches the one
obtained by
Ben-Or and Hassidim~\cite{BenOr-hassidim:2008:noisy} for the line.

%%%%%%%%%%%%%%%%%%%%%%%%%%%%%%%%%%%%%%%%%%%%%%%%%%

%% Hardness Results

\section{Hardness Results}\label{sec:negative}
We next turn to the question of
determining the optimal number of
queries required to find a target, or
--- equivalently --- the
decision problem of determining whether \TOTQ queries
are sufficient in the worst case.
All our hardness results
apply already for the noise-free version,
i.e., the special case where $p = 1$;
they are proved in Appendix~\ref{sec:hardness-proofs}.
For undirected graphs,
the result from Section~\ref{sec:deterministic}
implies a \quasip-time algorithm,
which makes the problem unlikely to
be \Hardness{NP}{hard}.
However, we show hardness results for this problem
based on two well-known conjectures.

\begin{conjecture*}{ETH}
The Exponential Time Hypothesis (ETH)
\cite{impagliazzo-paturi:2001:eth}
states that there exists
a positive constant $\hat{s}$ such that
$3$-CNF-SAT is not solvable in
time $2^{s n} \cdot \poly(n, m)$,
for any constant $s < \hat{s}$.
Here, $n$ and $m$ are the number of variables
and the number of clauses, respectively.
\end{conjecture*}

\begin{conjecture*}{SETH}
The Strong Exponential Time Hypothesis (SETH)
\cite{impagliazzo-paturi:2001:eth,calabro-impagliazzo-paturi:2009:seth} 
states that CNF-SAT does not admit
any algorithm with running time $2^{(1 - \epsilon) n} \cdot \poly(n, m)$
for any constant $\epsilon > 0$.
\end{conjecture*}

\begin{theorem}\label{thm:eth-hardness}
The problem of deciding,
given a connected undirected graph $G = (V, E)$
and integer \TOTQ,
whether there is an adaptive strategy
that finds the target in at most \TOTQ queries

%\begin{enumerate}
(1) %\item 
does not admit any algorithm
with $m^{o(\log n)}$  running time
unless the ETH is false.

(2) %\item 
does not admit any algorithm
with $O(m^{(1 - \epsilon)\log n})$
running time
for any constant $\epsilon > 0$
unless the SETH is false.

%\end{enumerate}
\end{theorem}

Our next hardness result is that the decision problem is
\Hardness{PSPACE}{complete} for \emph{directed} graphs. 
In fact, we will show a more fine-grained characterization
for the following ``semi-adaptive'' version of the problem,
even for undirected graphs:

\begin{problem}\label{pro:non-adaptive}
Given a connected undirected or
strongly connected directed graph $G$, 
and integers \NumRound and \EachRound,
the algorithm queries \EachRound vertices in each round
$1, 2, \ldots, \NumRound$. 
In response to the \EachRound queries, 
prior to the next round, it receives one edge out of each queried
vertex, lying on a shortest path from that vertex to the target.
The decision problem \TSSA{\NumRound}{\EachRound} is
to ascertain whether the
target can be found using \NumRound rounds of \EachRound queries each.
\end{problem}

When $\EachRound = 1$ and \NumRound is part of the input,
we obtain the original problem. When $\NumRound = 1$
(i.e., under the non-adaptive query model),
a straightforward reduction from
\textsc{Set Cover} establishes \Hardness{NP}{hardness}.
Our complexity-theoretic results
for \TSSA{\NumRound}{\EachRound} are the following:

\begin{theorem}\label{thm:directed-weak}
\begin{enumerate}
\item When the graph is directed
and \NumRound is part of the input,
\TSSA{\NumRound}{1} is \Hardness{PSPACE}{complete}.
\item When $\NumRound \geq 3$ is a constant
but \EachRound is part of the input, 
\TSSA{\NumRound}{\EachRound} is contained in $\Sigma_{2\NumRound-1}$
in the polynomial hierarchy,
and $\Sigma_{2\NumRound-5}$-hard.
This holds even for undirected graphs.
\end{enumerate}
\end{theorem}

%% reference to an appendix
%The proofs of Theorems~\ref{thm:eth-hardness}
%and \ref{thm:directed-weak} are given in
%Appendix~\ref{sec:hardness-proofs}.
The proof consists of a reduction from the
\textsc{Quantified Boolean Formulas (QBF)} problem and its variant
with a constant depth of quantifier nesting.

The graphs constructed in the $\Sigma_{2\NumRound-5}$-hardness proof
are undirected, but the \Hardness{PSPACE}{hardness} proof crucially
relies on long directed cycles to build a threat of a large number
of required queries, while avoiding an exponential blowup in the
reduction.\footnote{
Indeed, long cycles are the main obstacle
to efficient \binarysearch: 
in Appendix~\ref{sec:almost-undirected},
we show that the positive result from Section~\ref{sec:deterministic}
can be extended to strongly connected directed graphs which are almost
undirected in the sense that each edge appears in a short cycle.}

If node queries can incur non-uniform cost, then the same effect can
be obtained for undirected graphs, giving us the following theorem
%% reference to an appendix
%(whose proof is also given in Appendix~\ref{sec:hardness-proofs}):

\begin{theorem}\label{thm:undirected-cost}
It is \Hardness{PSPACE}{complete} to determine,
given a connected and unweighted
undirected graph $G = (V, E)$ 
with vertex cost function $\COSTFUNCTION: V \rightarrow Z_{+}$ 
and an integer \TOTQ, whether there is an adaptive query strategy
which always finds the target while paying a total cost of at most \TOTQ.
This hardness holds even when $G$ has diameter at most $13$.
\end{theorem}

Notice that for any connected and unweighted
undirected graph with diameter
\MAXDIST, the target can always be found using at most $\MAXDIST - 1$
queries, as follows: 
start with an arbitrary node; whenever the edge 
$e = (q, v)$ is revealed in response to a query $q$, the next
vertex to query is $v$.
If the optimal number of queries is upper-bounded by a constant,
then the exhaustive search outlined in Section~\ref{sec:deterministic} needs
only polynomial time to find an optimal querying strategy. 
Thus, Theorem~\ref{thm:undirected-cost} shows that non-uniform costs
significantly change the problem.

%%%%%%%%%%%%%%%%%%%%%%%%%%%%%%%%%%%%%%%%%%%%%%%%%%

%% Conclusion and Open Problems

\section{Conclusion and Open Problems} \label{sec:conclusion}

We proved that using vertex queries (which reveal either that
the target is at the queried vertex, or exhibit an edge on
a shortest path to the target), a target can always be found in an
undirected positively weighted graph using at most $\log n$ queries. 
When queries give the correct answer only with probability
$p > \Half$, we obtain an (essentially information-theoretically optimal)
algorithm which finds the target with probability at least
$1 - \ErrProb$ using
$(1 - \ErrProb) \log n /(1 - H(p)) + o(\log n) + O(\log^2(1/\ErrProb))$
queries in expectation.
We also proved several hardness results for
different variants, ranging from hardness under the ETH and SETH to
\Hardness{PSPACE}{completeness}.

Our work raises a number of immediate technical questions
for future work, but also some broader directions.
Some of the more immediate directions are:
\begin{enumerate}
\item In light of our hardness results for finding the exact answer,
how hard is it to find an \emph{approximately} optimal strategy?
Can this be done in polynomial time?

\item For directed graphs which are almost undirected,
close the gap
(see Appendix~\ref{sec:almost-undirected})
between the upper bound of 
$\MAXCYCLE \ln(2) \cdot \log n$, and the lower bound of
$\frac{\MAXCYCLE-1}{\log \MAXCYCLE} \cdot \log n$ on the number of
queries required.
Clearly, if $\MAXCYCLE \in \omega(n/\log n)$,
the upper bound cannot be tight, and we conjecture that
the lower bound may be asymptotically tight.

\item For undirected graphs, is the problem of deciding if an
adaptive strategy using at most \TOTQ queries exists in \Complexity{NP}?
What about \Complexity{Co-NP}?

\item For the semi-adaptive version, close the gap between the 
$\Sigma_{2\NumRound-5}$-hardness and the membership in 
$\Sigma_{2\NumRound-1}$. We believe that a more involved hardness
proof establishes $\Sigma_{2\NumRound-3}$-hardness,
which still leaves a gap with the upper bound.
\end{enumerate}

The semi-adaptive version of the problem raises a number of
interesting questions. In preliminary work, we have shown that
for the \NumDim-dimensional hypercube (with $n = 2^\NumDim$ vertices),
an algorithm using 2 queries per round can
use significantly fewer rounds than one using only a single query per round
(namely, $\frac{\log n}{2}$ instead of $\log n$).
Surprisingly, moving from 2 to 3 or more queries per round does
not give any further significant improvement; the next improvement
in the worst-case number of rounds arises when the number of queries
per round is $\EachRound = \Omega(\log \NumDim)$.
This fact is based on a combinatorial result of
Kleitman and Spencer~\cite{kleitman-spencer:1973:independent-sets}.
This observation raises the question
how a larger number of queries per round affects
the number of rounds required.
In particular, 
when the worst case over all connected undirected graphs is
considered,
will a larger number of queries per round
always guarantee a multiplicatively smaller number of rounds?

As discussed in the introduction, we can think of a positively
edge-weighted graph as a finite metric space.
This naturally suggests a generalization to infinite metric spaces. 
For instance, we may try to solve a puzzle such as: 
``Which major city in the United States lies
straight East of San Diego,
and straight Southeast of Seattle?\footnote{Dallas, TX.}''
More generally, for an arbitrary metric space $d$, 
the ``direction'' in which the target lies
from the queried point $q$ can be characterized by revealing
a point $v \neq q$ such that $d(q,v) + d(v,t) = d(q,t)$.

When the metric space is $\R^\NumDim$ with $\NumDim \geq 2$,
endowed with the $L_p$ norm for $p \in (1, \infty)$,  
it suffices to query two points $q_1, q_2$
not collinear with the target:
their query responses reveal two lines on which $t$ must lie,
and these lines intersect in a unique point.
For the $L_1$ norm in $\NumDim \geq 2$ dimensions,
on the other hand,
an adversary can always provide only answers
that get closer to the target in the first dimension,
without ever revealing any information about the second. 
Hence, no adaptive algorithm can provide any guarantees.

For more general metric spaces, e.g., arbitrary surfaces, the question
of how many queries are required to locate a point is wide open. 
In preliminary work, we have shown that in an
arbitrary polygon with $n$ corners, a target can always be found using
$O(\log n)$ queries. However, our approach does not extend to higher
genus.

There are many other fundamental generalizations of the model
that are possible, and would be of interest. 
These include querying an approximate median instead of an exact
one, receiving answers that lie on approximately shortest paths, or
searching for multiple targets at once. Absent careful definitions,
some of these generalizations make the problem trivially impossible
to solve, but with suitable definitions,
challenging generalizations arise.

%%%%%%%%%%%%%%%%%%%%%%%%%%%%%%%%%%%%%%%%%%%%%%%%%%

%\newpage

\section{Acknowledgments}

We would like to thank Noga Alon, Sepehr Assadi, Yu Cheng,
Shaddin Dughmi, Bobby Kleinberg, Daniel Ro\ss, Shanghua Teng, Adam
Wierman, and Avi Wigderson for useful discussions and pointers.

%%%%%%%%%%%%%%%%%%%%%%%%%%%%%%%%%%%%%%%%%%%%%%%%%%

%\newpage

\bibliographystyle{abbrv}
\bibliography{../bibliography/names,../bibliography/conferences,../bibliography/references}

\begin{thebibliography}{10}

\bibitem{abboud-backurs-williams:2015-lcs}
A.~Abboud, A.~Backurs, and V.~Vassilevska~Williams.
\newblock Tight hardness results for {LCS} and other sequence similarity
  measures.
\newblock In {\em Proc. 56th IEEE Symp. on Foundations of Computer Science},
  2015.

\bibitem{aslam:1995:noisy-learning-searching}
J.~A. Aslam.
\newblock {\em Noise Tolerant Algorithms for Learning and Searching}.
\newblock PhD thesis, 1995.

\bibitem{BenAsher-farchi:1997:trees}
Y.~Ben-Asher and E.~Farchi.
\newblock The cost of searching in general trees versus complete binary trees.
\newblock Technical report, 1997.

\bibitem{BenAsher-farchi-newman:1999:trees-edge-poly}
Y.~Ben-Asher, E.~Farchi, and I.~Newman.
\newblock Optimal search in trees.
\newblock {\em SIAM J. on Computing}, 28(6):2090--2102, 1999.

\bibitem{BenOr-hassidim:2008:noisy}
M.~Ben-Or and A.~Hassidim.
\newblock The bayesian learner is optimal for noisy binary search (and pretty
  good for quantum as well).
\newblock In {\em Proc. 49th IEEE Symp. on Foundations of Computer Science},
  pages 221--230, 2008.

\bibitem{borgstrom-kosaraju:1993:comparison-error}
R.~S. Borgstrom and S.~R. Kosaraju.
\newblock Comparison-based search in the presence of errors.
\newblock In {\em Proc. 25th ACM Symp. on Theory of Computing}, pages 130--136,
  1993.

\bibitem{braverman-ko-weinstein:2015:nash}
M.~Braverman, Y.~K. Ko, and O.~Weinstein.
\newblock Approximating the best {Nash Equilibrium} in $n^{o(\log n)}$-time
  breaks the {Exponential Time Hypothesis}.
\newblock In {\em Proc. 26th ACM-SIAM Symp. on Discrete Algorithms}, pages
  970--982, 2015.

\bibitem{bringmann-kunnemann:2015:quadratic-lower-bound}
K.~Bringmann and M.~K{\"u}nnemann.
\newblock Quadratic conditional lower bounds for string problems and dynamic
  time warping.
\newblock In {\em Proc. 56th IEEE Symp. on Foundations of Computer Science},
  2015.

\bibitem{burnashev-zigangirov:1974:estimation}
M.~V. Burnashev and K.~S. Zigangirov.
\newblock An interval estimation problem for controlled observations.
\newblock {\em Problemy Peredachi Informatsii}, 10:51--61, 1974.

\bibitem{calabro-impagliazzo-paturi:2009:seth}
C.~Calabro, R.~Impagliazzo, and R.~Paturi.
\newblock The complexity of satisfiability of small depth circuits.
\newblock In {\em Proc. of IWPEC 2009}, volume 5917 of {\em Lecture Notes in
  Computer Science}, pages 75--85, 2009.

\bibitem{carmo-donadelli-kohayakawa-laber:2004:poset}
R.~Carmo, J.~Donadelli, Y.~Kohayakawa, and E.~S. Laber.
\newblock Searching in random partially ordered sets.
\newblock {\em Theoretical Computer Science}, 321(1):41--57, 2004.

\bibitem{cicalese-jacobs-laber-molinaro:2011:tree-edge-average}
F.~Cicalese, T.~Jacobs, E.~Laber, and M.~Molinaro.
\newblock On the complexity of searching in trees and partially ordered
  structures.
\newblock {\em Theoretical Computer Science}, 412(50):6879--6896, 2011.

\bibitem{cicalese-jacobs-laber-valentim:2012:tree-edge-cost}
F.~Cicalese, T.~Jacobs, E.~Laber, and C.~Valentim.
\newblock The binary identification problem for weighted trees.
\newblock {\em Theoretical Computer Science}, 459:100--112, 2012.

\bibitem{cicalese-mundici-vaccaro:2002:semi-adaptive}
F.~Cicalese, D.~Mundici, and U.~Vaccaro.
\newblock Least adaptive optimal search with unreliable tests.
\newblock {\em Theoretical Computer Science}, 270(1--2):877--893, 2002.

\bibitem{dereniowski:2008:edge-ranking}
D.~Dereniowski.
\newblock Edge ranking and searching in partial orders.
\newblock {\em Discrete Applied Mathematics}, 156(13):2493--2500, 2008.

\bibitem{dhagat-gacs-winkler:1992:twenty-question-lier}
A.~Dhagat, P.~G\'{a}cs, and P.~Winkler.
\newblock On playing ``twenty questions'' with a liar.
\newblock In {\em Proc. 3rd ACM-SIAM Symp. on Discrete Algorithms}, pages
  16--22, 1992.

\bibitem{farrell:1964:sample-size}
R.~H. Farrell.
\newblock Asymptotic behavior of expected sample size in certain one sided
  tests.
\newblock {\em Annals of Mathematical Statistics}, 35(1):36--72, 1964.

\bibitem{feige-raghavan-peleg-upfal:1994:noisy}
U.~Feige, P.~Raghavan, D.~Peleg, and E.~Upfal.
\newblock Computing with noisy information.
\newblock {\em SIAM J. on Computing}, 23(5):1001--1018, 1994.

\bibitem{hoeffding:1963:probability}
W.~Hoeffding.
\newblock Probability inequalities for sums of bounded random variables.
\newblock {\em Journal of the American statistical association},
  58(301):13--30, 1963.

\bibitem{horstein:1963:noiseless-feedback}
M.~Horstein.
\newblock Sequential transmission using noiseless feedback.
\newblock {\em IEEE Trans. Inf. Theor.}, 9(3):136--143, 1963.

\bibitem{impagliazzo-paturi:2001:eth}
R.~Impagliazzo and R.~Paturi.
\newblock On the complexity of $k$-{SAT}.
\newblock {\em Journal of Computer and System Sciences}, 62:367--375, 2001.

\bibitem{iyer-ratliff-vijayan:1988:node-ranking}
A.~V. Iyer, H.~D. Ratliff, and G.~Vijayan.
\newblock Optimal node ranking of trees.
\newblock {\em Information Processing Letters}, 28(5):225--229, 1988.

\bibitem{jedynak-frazier-sznitman:2012:noise-entropy-loss}
B.~Jedynak, P.~I. Frazier, and R.~Sznitman.
\newblock Twenty questions with noise: Bayes optimal policies for entropy loss.
\newblock {\em Journal of Applied Probability}, 49(1):114--136, 2012.

\bibitem{jordan:1869:assemblages}
C.~Jordan.
\newblock Sur les assemblages de lignes.
\newblock {\em Journal f\"{u}r die reine und angewandte Mathematik},
  70:185--190, 1869.

\bibitem{karp-kleinberg:2007:noisy}
R.~M. Karp and R.~Kleinberg.
\newblock Noisy binary search and its applications.
\newblock In {\em Proc. 18th ACM-SIAM Symp. on Discrete Algorithms}, pages
  881--890, 2007.

\bibitem{kleitman-spencer:1973:independent-sets}
D.~J. Kleitman and J.~Spencer.
\newblock Families of $k$-independent sets.
\newblock {\em Discrete Mathematics}, 6, 1973.

\bibitem{laber-milidiu-pessoa:2001:binary-search-cost}
E.~S. Laber, R.~L. Milidi\'{u}, and A.~A. Pessoa.
\newblock On binary searching with non-uniform costs.
\newblock In {\em Proc. 12th ACM-SIAM Symp. on Discrete Algorithms}, pages
  855--864, 2001.

\bibitem{lam-yue:1998:edge-ranking-hard}
T.~W. Lam and F.~L. Yue.
\newblock Edge ranking of graphs is hard.
\newblock {\em Discrete Applied Mathematics}, 85(1):71--86, 1998.

\bibitem{lam-yue:2001:edge-ranking-linear}
T.~W. Lam and F.~L. Yue.
\newblock Optimal edge ranking of trees in linear time.
\newblock {\em Algorithmica}, 30(1):12--33, 2001.

\bibitem{linial-saks:1985:searching}
N.~Linial and M.~Saks.
\newblock Searching ordered structures.
\newblock {\em Journal of Algorithms}, 6(1):86--103, 1985.

\bibitem{mozes-onak-weimann:2008:trees-edge-linear}
S.~Mozes, K.~Onak, and O.~Weimann.
\newblock Finding an optimal tree searching strategy in linear time.
\newblock In {\em Proc. 19th ACM-SIAM Symp. on Discrete Algorithms}, pages
  1096--1105, 2008.

\bibitem{naghshvar-javidi-chaudhuri:2015:active-learning-noise}
M.~Naghshvar, T.~Javidi, and K.~Chaudhuri.
\newblock Bayesian active learning with non-persistent noise.
\newblock {\em IEEE Transactions on Information Theory}, 61(7):4080--4098,
  2015.

\bibitem{nowak:2009:noisy-binary-search}
R.~Nowak.
\newblock Noisy generalized binary search.
\newblock In {\em Proc. 23rd Advances in Neural Information Processing
  Systems}, pages 1366--1374, 2009.

\bibitem{onak-parys:2006:trees-vertex-linear}
K.~Onak and P.~Parys.
\newblock Generalization of binary search: Searching in trees and forest-like
  partial orders.
\newblock In {\em Proc. 47th IEEE Symp. on Foundations of Computer Science},
  pages 379--388, 2006.

\bibitem{papadimitriou:1994:complexity}
C.~Papadimitriou.
\newblock {\em Computational Complexity}.
\newblock Addison-Wesley, 1994.

\bibitem{pedrotti:1999:costant-rate-lies}
A.~Pedrotti.
\newblock Searching with a constant rate of malicious lies.
\newblock In {\em Proceedings of the International Conference on Fun with
  Algorithms (FUN-98)}, pages 137--147, 1999.

\bibitem{pelc:1989:error}
A.~Pelc.
\newblock Searching with known error probability.
\newblock {\em Theoretical Computer Science}, 63(2):185--202, 1989.

\bibitem{pelc:2002:games-errors}
A.~Pelc.
\newblock Searching games with errors --- fifty years of coping with liars.
\newblock {\em Theoretical Computer Science}, 270(1--2):71--109, 2002.

\bibitem{renyi:1961:problem}
A.~R\'{e}nyi.
\newblock On a problem of information theory.
\newblock {\em MTA Mat.Kut.Int.Kozl.}, 6 B:505--516, 1961.

\bibitem{rivest-meyer-kleitman-Winklmann-Spencer:1980:coping-error}
R.~L. Rivest, A.~R. Meyer, D.~J. Kleitman, K.~Winklmann, and J.~Spencer.
\newblock Coping with errors in binary search procedures.
\newblock {\em Journal of Computer and System Sciences}, 20(3):396--404, 1980.

\bibitem{schaffer:1989:node-ranking-linear}
A.~A. Sch{\"a}ffer.
\newblock Optimal node ranking of trees in linear time.
\newblock {\em Information Processing Letters}, 33(2):91--96, 1989.

\bibitem{ulam:1991:adventures-mathematician}
S.~M. Ulam.
\newblock {\em Adventures of a Mathematician}.
\newblock 1991.

\bibitem{waeber-frazier-henderson:2013:noisy-responses}
R.~Waeber, P.~I. Frazier, and S.~G. Henderson.
\newblock Bisection search with noisy responses.
\newblock {\em SIAM Journal on Control and Optimization}, 51(3):2261--2279,
  2013.

\bibitem{yan-chaudhuri-javidi:2015:learning-noisy}
S.~Yan, K.~Chaudhuri, and T.~Javidi.
\newblock Active learning from noisy and abstention feedback.
\newblock In {\em Allerton Conference on Communication, Control and Computing},
  2015.

\end{thebibliography}

%%%%%%%%%%%%%%%%%%%%%%%%%%%%%%%%%%%%%%%%%%%%%%%%%%

\newpage

\appendix

%% Hardness Proofs

\section{Hardness Proofs} \label{sec:hardness-proofs}

Here, we supply the hardness proofs omitted
from Section~\ref{sec:negative}.
We begin with the proof for Theorems~\ref{thm:directed-weak}
and \ref{thm:undirected-cost}. 
The proofs are quite similar, so we present them together.

%\eedelete: The theorems are restated here for convenience.
\Eat{
\begin{rtheorem}{Theorem}{\ref{thm:directed-weak}}
\begin{enumerate}
\item When \NumRound is part of the input,
even \TSSA{\NumRound}{1} is \Hardness{PSPACE}{complete}
for directed graphs $G$.
\item When $\NumRound \geq 3$ is a constant
but \EachRound is part of the input, 
\TSSA{\NumRound}{\EachRound} is contained in $\Sigma_{2\NumRound-1}$
in the polynomial hierarchy,
and $\Sigma_{2\NumRound-5}$-hard.
This holds even when the graph $G$ is undirected.
\end{enumerate}
\end{rtheorem}

\begin{rtheorem}{Theorem}{\ref{thm:undirected-cost}}
It is \Hardness{PSPACE}{complete} to determine,
given a connected and unweighted
undirected graph $G = (V, E)$ 
with vertex cost function $\COSTFUNCTION: V \rightarrow Z_{+}$ 
and an integer \TOTQ, whether there is an adaptive query strategy
which always finds the target while paying a total cost of at most \TOTQ.
This hardness holds even when $G$ has diameter at most $9$.
\end{rtheorem}
}

%\begin{proof}
First, recall (see \cite[Chapter 17.2]{papadimitriou:1994:complexity}
for a more detailed discussion)
that $\Sigma_{2\NumRound - 1}$ (here, we are interested
only in odd classes) can be characterized as
all problems that can be expressed as
\[
\exists \vc{x}_1 \forall \vc{y}_1 \exists \vc{x}_2 \forall \vc{y}_2
\cdots \exists \vc{x}_\NumRound:
F(\vc{x}_1, \vc{y}_1, \vc{x}_2, \vc{y}_2, \ldots, \vc{x}_\NumRound),
\]
where each $\vc{x}_i = (x_{i,1}, \ldots, x_{i,\EachRound}),
\vc{y}_i = (y_{i,1}, \ldots, y_{i,\EachRound})$ 
is a vector of \EachRound Boolean variables, 
and $F$ is a predicate that can be evaluated in polynomial time.
A canonical hard problem for $\Sigma_{2\NumRound-1}$, which we will
use for the hardness proof here, is the problem
\GQBF{\NumRound}, in which we specifically choose 
$F(\vc{x}_1, \vc{y}_1, \ldots, \vc{x}_\NumRound) = \bigwedge_{\ell = 1}^m \Clause{\ell}$
as a CNF formula.
Without loss of generality, we may assume that no clause contains
both a variable and its negation, since such a clause is always
satisfied, and can thus be pruned.
When $\NumRound$ is part of the input, \GQBF{\NumRound} is equivalent
to the well-known QBF problem --- in that case,
the problem is \Hardness{PSPACE}{hard}
even when each $\vc{x}_i, \vc{y}_i$ is just a single Boolean variable.

\GQBF{\NumRound} (and thus also QBF) can be considered as a two-player
game in which the players take turns choosing values for the variables.
During the \Ordinal{i} round, the first player assigns
either \true or \false to each of the \EachRound variables in $\vc{x}_i$;
subsequently, the second player chooses Boolean values for the
$\EachRound$ variables in $\vc{y}_i$. 
The first player's objective is to satisfy $F$ while the second player
wants $F$ to become \false.
Using this interpretation, \GQBF{\NumRound} can be rephrased as follows:
``Is there a winning strategy for the first player?''

\begin{extraproof}{Theorem \ref{thm:directed-weak}}
We consider the adaptive query problem as a two-player game,
in which the first player (whom we call the \firstplayer)
queries \EachRound vertices at a time, while the second player (the
\secondplayer), for each queried vertex $q$, chooses an
outgoing edge from $q$ lying on a shortest path from $q$ to the target. 
The \firstplayer wins if after \NumRound rounds, he can uniquely
identify the target vertex based on the responses, while the
\secondplayer wins if after \NumRound rounds of \EachRound queries
each, there is still more than one potential target vertex
consistent with all answers.

To prove membership in $\Sigma_{2\NumRound - 1}$ (and thus also in
\Complexity{PSPACE} when \NumRound is part of the input), first
notice that the \firstplayer's choice of \EachRound nodes to query
can be encoded in $\EachRound \log(n)$ bits, i.e., Boolean variables,
as can the \secondplayer's response of \EachRound edges.
Membership in $\Sigma_{2\NumRound}$ would now be obvious,
as we could quantify over all of the \secondplayer's responses,
and given all responses, it is obvious
how to decide whether the target is
uniquely identified.

To prove the stronger bound of $\Sigma_{2\NumRound - 1}$, 
consider the following decision problem:
given all queries, as well as possibly responses to some of them, 
can the \secondplayer respond to the remaining queries while
keeping at least two candidate targets?
This question can be decided as follows.
The \secondplayer enumerates all target pairs $\Set{t_1, t_2}$.
For each query independently, she checks if there is a
response that would be consistent with both $t_1$ and $t_2$.
If such responses exist, then the \secondplayer can ensure that the
\firstplayer cannot differentiate between $t_1$ and $t_2$; otherwise,
the \firstplayer will succeed. By exhaustively checking all pairs,
the \secondplayer tries all possible winning strategies.

Let $\NumExists = \NumRound - 2$.
To prove hardness, we reduce from \GQBF{2\NumExists - 1},
i.e., we let $\NumExists$ be the number of $\exists$ quantifiers
in the formula.
For the \Hardness{PSPACE}{hardness} proof, we only need each
$\vc{x}_i, \vc{y}_i$ to be a single variable.

Given an instance of \GQBF{2\NumExists - 1}, 
in which each $\vc{x}_i, \vc{y}_i$ has (without loss of generality)
exactly \EachRound Boolean variables
($\EachRound = 1$ for the \Hardness{PSPACE}{hardness} reduction),
we construct an unweighted strongly connected graph 
$G = (V, E)$ with the following pieces,
illustrated in Figure~\ref{fig:reduction}.

\InsertFigure{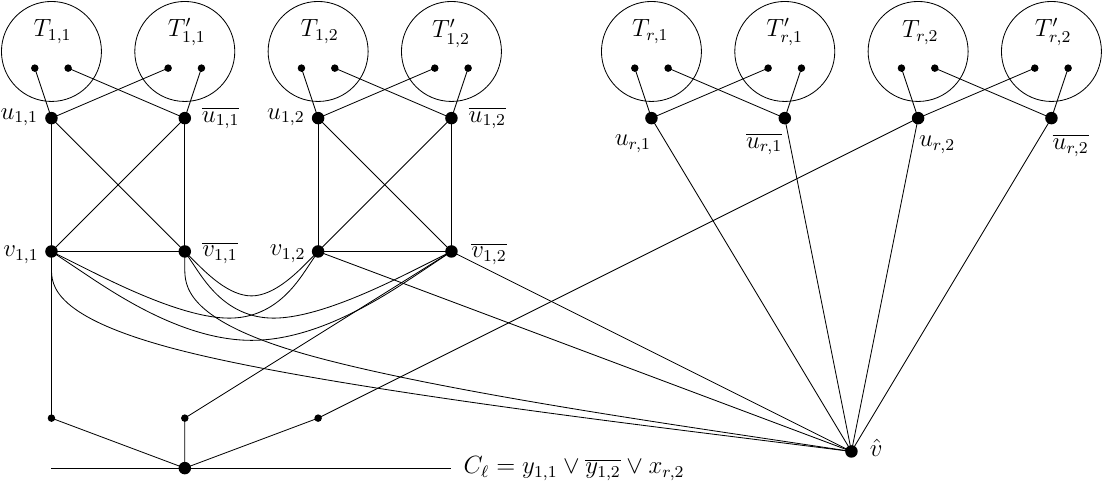}{\columnwidth}%
{A schematic depiction of the graphs produced by the reduction.}

\begin{itemize}
\item For each variable $x_{i, j}$,
add two \emph{literal vertices of type \RNum{1}},
named $u_{i, j}$ and $\Not{u_{i, j}}$.
Similarly, corresponding to each variable $y_{i, j}$,
add two \emph{literal vertices of type \RNum{2}},
named $v_{i, j}$ and $\Not{v_{i, j}}$. 
For every $1 \leq i < \NumExists$
and $1 \leq j \leq \EachRound$,
add undirected edges
from both $u_{i, j}$ and $\Not{u_{i, j}}$
to both $v_{i, j}$ and $\Not{v_{i, j}}$.
Add undirected edges between each pair of literal vertices
of type \RNum{2}.

Furthermore, add an extra vertex \ExtraNode.
Add undirected edges between \ExtraNode
and each of the literal vertices of type \RNum{2},
as well as between \ExtraNode
and both $u_{\NumExists, j}$ and $\Not{u_{\NumExists, j}}$
for all $1 \leq j \leq \EachRound$.
For the \Hardness{PSPACE}{hardness} reduction,
also add directed edges from \ExtraNode to both
$u_{i,j}$ and $\Not{u_{i,j}}$ for $i < \NumExists$;
however, for the \Hardness{$\Sigma_{2\NumExists - 1}$}{hardness},
there is no edge between \ExtraNode and
$u_{i,j}$ or $\Not{u_{i,j}}$ for $i < \NumExists$.

\item For each $1 \leq i \leq \NumExists,
1 \leq j \leq \EachRound$,
add two \emph{critical gadgets} \CritGadgetA{i,j} and \CritGadgetB{i,j}
that will be specified momentarily.
Choose two arbitrary distinct vertices
in each critical gadget as \emph{critical nodes}. 
Let \CritNodeA{i,j} and \Not{\CritNodeA{i,j}}
be critical nodes of \CritGadgetA{i,j},
and \CritNodeB{i,j} and \Not{\CritNodeB{i,j}} be in \CritGadgetB{i,j}.
Add undirected edges between $u_{i, j}$
and both \CritNodeA{i,j} and \CritNodeB{i,j},
as well as between $\Not{u_{i, j}}$
and both \Not{\CritNodeA{i,j}} and \Not{\CritNodeB{i,j}}.

For the \Hardness{PSPACE}{hardness} reduction, the \Ordinal{i}
critical gadgets are directed cycles of length $\NumExists + 3 - i$;
for the \Hardness{$\Sigma_{2\NumExists - 1}$}{hardness} proof, 
they are undirected paths of length
$2(\EachRound + 1)^{\NumExists + 2 - i} - 1$.

\item Corresponding to each clause \Clause{\ell} in $F$,
add two \emph{clause gadgets} \ClGadgetA{\ell} and \ClGadgetB{\ell}.
Each clause gadget is an undirected path of length
$2(\EachRound + 1)^2-1$.
(For the \Hardness{PSPACE}{hardness} reduction, where $\EachRound = 1$,
the path has length 7.)
Let \ClNodeA{\ell} and \ClNodeB{\ell} be vertex number
$2(\EachRound+1)$ on the respective paths.
For the case $\EachRound = 1$, this is the midpoint of the path, and
more generally, the leftmost of the set of vertices dividing the path
into $\EachRound+1$ equal-sized subpaths.

\item Whenever a literal ($x_{i, j}$,
$\Not{x_{i, j}}$, $y_{i, j}$, $\Not{y_{i, j}}$)
appears in clause \Clause{\ell}, add two \emph{intermediate nodes}
specific to this pair of a literal and clause.
Add undirected edges between the intermediate nodes and
the corresponding literal vertex
($u_{i, j}$, $\Not{u_{i, j}}$, $v_{i, j}$, $\Not{v_{i, j}}$)
and add an undirected edge between the first intermediate node and
\ClNodeA{\ell}, and between the second intermediate node and \ClNodeB{\ell}.
(The intermediate nodes connect the corresponding literal
vertex to both \ClNodeA{\ell} and \ClNodeB{\ell}
via paths of length $2$.) See Figure~\ref{fig:reduction}.
\end{itemize}

Recall that
we set the total number of rounds to be $\NumRound = \NumExists + 2$.
The \firstplayer is allowed to query \EachRound vertices in each round.
Notice that the graph we construct is strongly connected, because all
the critical gadgets and clause gadgets are strongly connected.
Moreover, although there are some directed edges in the construction for
\Hardness{PSPACE}{hardness}, all the edges are undirected
for the \Hardness{$\Sigma_{2\NumExists-1}$}{hardness} reduction.

Furthermore, while the size of the critical gadgets for the
$\Sigma_{2\NumExists - 1}$-hardness reduction grows exponentially in
\NumRound, it is polynomial for constant
\NumExists (and therefore, \NumRound).
The reduction thus always takes polynomial time.
We begin with a simple lemma elucidating the role of the critical
gadgets and clause gadgets.

\begin{lemma}[Gadgets] \label{lemma:gadgets}
\begin{enumerate}
\item Conditioned on knowing that the target is in
$\CritGadgetA{i,j}$ (or $\CritGadgetB{i,j}$), and nothing else,
$\NumExists + 2 - i = \NumRound - i$ rounds of \EachRound queries
are necessary and sufficient to find the target in the worst case.
\item Conditioned on knowing that the target is in a specific clause
gadget $\ClGadgetA{\ell}$ or $\ClGadgetB{\ell}$, but nothing more, 
2 more rounds of queries are necessary.
The same number is also sufficient to identify the target uniquely,
even when the target may be in an intermediate node adjacent to the
clause gadget.
\end{enumerate}
\end{lemma}

\begin{proof}
We begin with the first part of the lemma:
For the directed cycle, until all but one of the vertices have been
queried, there are always at least two candidates; querying all but
one vertex is also clearly sufficient.

For a path of length $M$, querying \EachRound nodes can at most
eliminate those nodes, and results at best in $\EachRound + 1$ subpaths
of length $(M-\EachRound)/(\EachRound+1)$ each.
Spacing the query nodes evenly matches this bound.
A simple proof by induction now establishes the bound.

The proof of the second part is analogous: 
The length of the path is chosen to require exactly 2 rounds of
\EachRound queries.
The choice of the vertices \ClNodeA{\ell} and \ClNodeB{\ell} 
ensures that if the target was at an intermediate vertex, this
fact is revealed in the first round of querying.
\end{proof}

We claim that the first player in the QBF game has a winning strategy
if and only if the \firstplayer can find any target in $G$ using at most
\NumRound rounds of queries.

%\begin{enumerate}
%\item 
(1) First, assume that there exists a strategy for finding
the target in $G$ using no more than \NumRound rounds of queries.

Let $Q_i$ be the set of vertices queried by the \firstplayer
in the \Ordinal{i} round.
We claim the following, for each $i \leq \NumExists$: 
if $Q_{i'} \subseteq \SpecSet{u_{i',j}, \Not{u_{i',j}}}{j=1,\ldots,\EachRound}$
for each $i' < i$, 
and all of the second player's responses were toward either
$v_{i',j}$ or $\Not{v_{i',j}}$, then 
$Q_i \subseteq \SpecSet{u_{i,j}, \Not{u_{i,j}}}{j=1, \ldots, \EachRound}$,
and furthermore, $Q_i$ contains exactly one of 
$\Set{u_{i,j}, \Not{u_{i,j}}}$ for each $j$.
The reason is that under the assumption about prior queries,
all critical gadgets \CritGadgetA{i,j} and \CritGadgetB{i,j} are still in $S$.
Having already used $i-1$ rounds previously, the \firstplayer has only
$\NumRound - (i-1)$ rounds left,
and $\NumRound - i$ rounds are
necessary in order to find a target in
one of these critical gadgets,
by the first part of Lemma~\ref{lemma:gadgets}.
Since the \secondplayer can also choose which of the $2\EachRound$
critical gadgets contains the target, the \firstplayer has to identify the
correct gadget using at most one round, which is only accomplished by
choosing one of $u_{i,j}, \Not{u_{i,j}}$ for each $j$.

We now define the following mapping from the \firstplayer's strategy
to a winning strategy for the first player in the formula game.
When the \firstplayer queries $u_{i,j}$, the first player assigns 
$x_{i,j} = \true$, whereas when the \firstplayer queries $\Not{u_{i,j}}$,
the first player sets $x_{i,j} = \false$. 
For $i < \NumExists$, in response to the first player's setting of $x_{i,j}$,
the second player will set $y_{i,j}$ either \true or \false.
If the second player sets $y_{i,j} = \true$, then we have the \secondplayer
reveal $\Not{v_{i,j}}$, whereas when
$y_{i,j} = \false$, the \secondplayer reveals $v_{i,j}$ instead.
By the previous claim, if $i < \NumExists$, the vertex player's next
queries must be to either $u_{i + 1,j}$ or $\Not{u_{i + 1,j}}$, meaning
that we can next set all of the $x_{i + 1,j}$ by the same procedure.
For the \Ordinal{\NumRound} round of queries, we make the \secondplayer
reveal the edge toward \ExtraNode.
We thus obtain a variable assignment to all $x_{i,j}$ and $y_{i,j}$,
and it remains to show that it satisfies all clauses \Clause{\ell}.

By assumption, having used
$\NumExists = \NumRound - 2$ rounds of queries, the vertex
player can always identify the target with the remaining 
$2$ rounds. 
Note that for each matching pair 
$(\ClGadgetA{\ell}, \ClGadgetB{\ell})$ of clause gadgets,
both \ClGadgetA{\ell} and \ClGadgetB{\ell} are entirely
in the candidate set, or both have been completely ruled out.
This is because all responses were pointing to $v_{i,j}$ or
$\Not{v_{i,j}}$, and were thus symmetric for both \ClGadgetA{\ell} and
\ClGadgetB{\ell}.
By the second part of Lemma~\ref{lemma:gadgets}, it would take at
least 2 rounds of queries to identify a node in a
known clause gadget, and one more to identify the correct gadget.
This is too many rounds of queries, so no clause gadgets can 
be in $S$, and all clause gadget instances must have been
ruled out previously.

This means that for each clause \Clause{\ell}, there must have been
$1 \leq i \leq \NumExists$
such that the responses to the \Ordinal{i} query
ruled out the target being in either \ClGadgetA{\ell} or \ClGadgetB{\ell}.
Suppose that the queried set $Q_i$ contained a node
$q_{i,j} \in \Set{u_{i,j}, \Not{u_{i,j}}}$,
and $v_{i,j}$ or $\Not{v_{i,j}}$ (or \ExtraNode if $i = \NumExists$) 
is the answer.
Then, the clause gadgets could have been ruled out in one of two ways:
\begin{itemize}
\item $q_{i,j}$ is connected to \ClGadgetA{\ell} and \ClGadgetB{\ell}
via intermediate nodes, 
i.e., there is a path of length $2$ from $q_{i,j}$ to 
\ClNodeA{\ell} and also \ClNodeB{\ell}.
These clause gadgets have been ruled out because the answer of the query
was not toward one of the intermediate nodes
(connected to \ClNodeA{\ell} or \ClNodeB{\ell}).
In this case, by definition,
the literal corresponding to $q_{i,j}$ (either $x_{i,j}$ or $\Not{x_{i,j}}$)
is in \Clause{\ell}, and because $q_{i,j}$ was queried,
the literal is set to \true. 
Thus, \Clause{\ell} is satisfied under the assignment.

\item An intermediate node connects $v_{i,j}$ to \ClGadgetA{\ell} (and
\ClGadgetB{\ell}) 
and the \secondplayer chose $\Not{v_{i,j}}$, or --- symmetrically ---
an intermediate node connects $\Not{v_{i,j}}$ 
to \ClGadgetA{\ell} (and \ClGadgetB{\ell})
and the \secondplayer chose $v_{i,j}$
Without loss of generality, assume the first case.
In that case, by definition, $y_{i,j} = \true$ (because $\Not{v_{i,j}}$ was
chosen), and by construction of the graph, $y_{i,j} \in \Clause{\ell}$.
Again, this ensures that \Clause{\ell} is satisfied under the assignment.
\end{itemize}
In summary, we have shown that all clauses are satisfied,
meaning that the first player in the formula game has won the game.

%\item 
(2) For the converse direction, assume that the first player
has a winning strategy in the formula game.
We will use this strategy
to construct a winning strategy for the \firstplayer in the target
search game.

We begin by considering a round $i \leq \NumExists$
in which the vertex player needs to
make a choice. 
Assume that for each round $i' < i$, 
$Q_{i'} \subseteq \SpecSet{u_{i',j},\Not{u_{i',j}}}{j=1,\ldots,\EachRound}$, 
and furthermore, $Q_{i'}$ contains exactly one of
$\Set{u_{i',j},\Not{u_{i',j}}}$ for each $j$.
Also assume that the \secondplayer responded
with either $v_{i',j}$ or $\Not{v_{i',j}}$ for each query of
$u_{i',j}$ or $\Not{u_{i',j}}$.
(We will see momentarily that these assumptions are warranted.)
Interpret an answer pointing to $v_{i',j}$ as setting $y_{i',j} = \false$,
and an edge pointing to $\Not{v_{i',j}}$ as setting $y_{i',j} = \true$.
Consider the choice for $x_{i,j}$ prescribed by the first player's assumed
winning strategy, based on the history so far.
The vertex player will query $u_{i,j}$ if $x_{i,j} = \true$,
and $\Not{u_{i,j}}$ if $x_{i,j} = \false$.
We distinguish several cases, based on the response:
\begin{itemize}
\item If one of the queried vertices is the target, then clearly, the
  \firstplayer has won.

\item If the \secondplayer's response is toward an instance of
critical gadgets, then the target is known to lie in that
critical gadget. 
By Lemma~\ref{lemma:gadgets}, there exists a query strategy for
\CritGadgetA{i} (or \CritGadgetB{i}) 
which can find the target using at most $\NumRound - i$ rounds of
\EachRound queries. 
Together with the $i$ rounds of queries already used by the \firstplayer,
this gives a successful identification with at most \NumRound
rounds of queries total.

\item If the answer is toward an intermediate node connected to one of
the queried nodes, then the target must lie in the corresponding
clause gadget, say \ClGadgetA{\ell},
or is one of the intermediate nodes connected to this gadget.
By Lemma~\ref{lemma:gadgets}, it takes at most 2 more
rounds of queries to identify the target; together with the first $i$
rounds of queries, this is a successful identification with at most
\NumRound rounds of queries.

\item This leaves the case when the \secondplayer
chooses the edge toward $v_{i,j}$ or $\Not{v_{i,j}}$
(or \ExtraNode, if $i = \NumExists$),
justifying the assumption made earlier
that for each of the first $i$ rounds of queries,
the \secondplayer responds by revealing edges
toward either $v_{i,j}$ or $\Not{v_{i,j}}$.
\end{itemize}

In summary, the fact that the assignment satisfies all \Clause{\ell} implies
that the target cannot lie on any clause gadget.
The fact that the \secondplayer responded with edges toward $v_{i,j}$
or $\Not{v_{i,j}}$ to the first \NumExists rounds of queries implies
that the target cannot lie on any critical gadgets, either.

The remaining case is the \Ordinal{\NumExists} round
of queries, when the \secondplayer's response may contain edges toward
\ExtraNode.
Then, the only candidate nodes remaining after \NumRound rounds of queries 
are (some of) the literal vertices, (some of) the intermediate nodes
and \ExtraNode. 
In this case, $2$ more rounds are sufficient
for both of the reductions, as follows.
For the \Hardness{PSPACE}{hardness} reduction,
query \ExtraNode, and subsequently query the vertex
$u_{i, j}$, $\Not{u_{i, j}}$, $v_{i,j}$ or $\Not{v_{i,j}}$
with which the \secondplayer responds.
(Recall that there are edges from \ExtraNode to all the literal vertices.)
That query either reveals the target, or points to 
an intermediate vertex which is then known to be the target.
For the \Hardness{$\Sigma_{2\NumExists - 1}$}{hardness} reduction,
\ExtraNode is queried in the first extra round.
Next, the vertex player simultaneously queries whichever of
$v_{i, j}, \Not{v_{i, j}}$ the \secondplayer responded with, 
as well as the node of $u_{i, j}, \Not{u_{i, j}}$ that he has not
queried yet.
This will either reveal the target in one of the queried nodes, or
reveal an edge to an intermediate vertex which is then known to be the
target.
%\end{enumerate}

Finally, notice that the number of rounds was chosen to be
$\NumRound = \NumExists + 2$;
expressing the hardness result in terms of the number
of rounds of the target search game implies the claimed
$\Sigma_{2\NumRound - 5}$-hardness for the semi-adaptive version with
\NumRound rounds. 
We believe that a somewhat more complex construction and argument
improves this bound to $\Sigma_{2\NumRound - 3}$-hardness.
%, which however still does not quite match the upper bound.
\end{extraproof}

\begin{extraproof}{Theorem~\ref{thm:undirected-cost}}
We only detail the changes in the gadgets of
Theorem~\ref{thm:directed-weak} here.
First, we set an upper bound of $\TOTQ = \NumExists + 3$ on the cost,
instead of the bound $\NumRound = \NumExists + 2$ on the number of rounds.
Recall that for critical gadgets \CritGadgetA{i} and \CritGadgetB{i},
we only need the following property:
$\TOTQ - i$ queries are necessary and sufficient in the worst case
to find the target in each of them.
We therefore let each \CritGadgetA{i} and \CritGadgetB{i} consist of
two vertices connected with an undirected edge.
For each of these nodes, the query cost is $\TOTQ - i$.
The new critical gadgets still satisfy Lemma~\ref{lemma:gadgets},
which was all that the proof required.

Clause gadgets are replaced with paths of length 15 (instead of paths
of length 7), meaning that instead of two queries, three queries are
now necessary and sufficient to identify a target in a clause gadget.
Also, \ClNodeA{\ell} and \ClNodeB{\ell}
are again the middle points of these paths.
Finally, in the new construction, no directed edges are inserted from 
\ExtraNode to the $u_{i,j}, \Not{u_{i,j}}$ for $i < \NumExists$.
The only difference is that now, when all critical and clause gadgets have
been ruled out after $\NumExists$ queries, it 
takes $3$ more queries to find the target.
The rest of the proof stays the same as before.

It is easy to check that the diameter of the resulting graph
is at most $13$.
\end{extraproof}

\bigskip

We now prove Theorem~\ref{thm:eth-hardness}.
Hardness results based on ETH and SETH
have appeared in several recent works: for example,
Braverman et al.~\cite{braverman-ko-weinstein:2015:nash}
proved a \quasip-time lower bound for approximating
the best Nash equilibrium, while
Abbould et al.~\cite{abboud-backurs-williams:2015-lcs} and
Bringmann et al.~\cite{bringmann-kunnemann:2015:quadratic-lower-bound}
rule out sub-quadratic algorithms for a family of classical
string problems (e.g., Longest Common Subsequence).

\begin{extraproof}{Theorem~\ref{thm:eth-hardness}}
Let $\bigwedge_{\ell = 1}^m \Clause{\ell}$
be an instance of CNF-SAT with $n$ variables
and $m$ clauses (with $m$ polynomial in $n$).
Without loss of generality,
assume that $n = k^2$ is a perfect square.
(If it is not, we can add $O(\sqrt{n})$ dummy variables
to make it a perfect square.) 
Partition the variables into $k$ batches of $k$ variables each,
labeled $x_{j,i}$.

The overall construction idea is similar to the proofs of
Theorems~\ref{thm:directed-weak} and \ref{thm:undirected-cost}, but
slightly easier.
The (unweighted and undirected) graph looks as follows this time:

\begin{itemize}
\item For each batch $j$ ($1 \leq j \leq k$),
and each assignment $\vc{a} \in \Set{0,1}^k$,
construct three vertices: an \emph{assignment vertex}
$v_{j,\vc{a}}$ and two \emph{intermediate vertices}
$u_{j,\vc{a}}, u'_{j,\vc{a}}$.
Add edges between
$v_{j,\vc{a}}$ and $u_{j,\vc{a}}$, and between
$v_{j,\vc{a}}$ and $u'_{j,\vc{a}}$.
Add two extra nodes $\ExtraNode, \ExtraNode'$, connected via
an edge. 
Moreover, connect \ExtraNode with all assignment vertices
$v_{j,\vc{a}}$.

\item For each batch $j$ ($1 \leq j \leq k$),
add two critical gadgets \CritGadgetA{j} and \CritGadgetB{j},
each a simple path of length $2^{k - j + 3} - 1$.
Let \CritNodeA{j} and \CritNodeB{j} be the middle points of
\CritGadgetA{j} and \CritGadgetB{j}, respectively.
Connect \CritNodeA{j} to the intermediate nodes $u_{j,\vc{a}}$ for all
$\vc{a}$, and \CritNodeB{j} to $u'_{j,\vc{a}}$ for all $\vc{a}$.
Hence, all assignment vertices are connected to the
corresponding critical gadgets via paths of length two.

\item Corresponding to each clause \Clause{\ell} in the formula,
add a clause gadget \ClGadgetA{\ell},
which is a simple path of length $7$.
For each assignment vertex $v_{j, \vc{a}}$,
if $\vc{a}$ satisfies \Clause{\ell},
add a new intermediate node $u''_{j,\vc{a},\ell}$,
and connect it to both $v_{j, \vc{a}}$ and
the middle node of \ClGadgetA{\ell}.
\end{itemize}

The overall outline of the proof is similar to (but simpler than) the
one of Theorems~\ref{thm:directed-weak} and \ref{thm:undirected-cost}.
The key idea is again that to have any chance of finding a target in
the critical gadgets, an adaptive strategy must pick exactly one
assignment vertex from each batch; otherwise, a target in a
  critical gadget could not be identified.
This allows us to establish a one-to-one
correspondence between adaptive strategies and variable assignments. 
Revealing an edge to a critical or clause gadget would give
the adaptive strategy an easy winning option, so one can show that
w.l.o.g., all responses are to $\ExtraNode$. 
This rules out all clause gadgets for clauses satisfied by the
$\vc{a}$ for the queried vertex $v_{j, \vc{a}}$.

In order to succeed in the final two rounds with 
$\ExtraNode, \ExtraNode'$, a number of unqueried $v_{j,\vc{a}}$
and many $u''_{j,\vc{a},\ell}$ still remaining,
an algorithm must have eliminated all of the clause
gadgets from consideration, which is accomplished only when all
clauses are satisfied. 
(Conversely, if all critical and clause gadgets have been eliminated,
the algorithm can next query $\ExtraNode$ and the $v_{j,\vc{a}}$ that
is returned as the response.)
Hence, a satisfying variable assignment exists if and only if $k + 2$
queries are sufficient, as captured by the following lemma:

\begin{lemma}
There exists an adaptive strategy to find the target
in the constructed graph within at most $k + 2$ queries
if and only if the CNF formula is satisfiable.
\end{lemma}

The constructed graph has $N = O(m k 2^{k})$ vertices
and $M = O(m k 2^{k})$ edges; thus $\log N = k + o(k)$.
Assume that some algorithm \AAA decides whether there exists any
adaptive strategy to find the target with $k + 2$ queries.
We would obtain the following complexity-theoretic consequences:

\begin{enumerate}
\item If the formula is a $3$-CNF-SAT formula,
and the running time of \AAA is $M^{o(\log N)} = M^{o(k)}$,
then the reduction would give us an algorithm for $3$-CNF-SAT
with running time $M^{o(k)} = 2^{o(n)}$,
which contradicts the ETH.

\item For general CNF-SAT instances,
if the running time of \AAA is 
$O(M^{(1 - \epsilon)\log N})$,
then the above reduction would solve CNF-SAT with running time 
\begin{align*}
O((m\sqrt{n}2^{\sqrt{n}})^{(1-\epsilon) \sqrt{n}})
= O(2^{(1 - \epsilon/2)n}),
\end{align*}
contradicting the SETH.
\end{enumerate}
\end{extraproof}

%% Almost Undirected Graphs

\section{Almost Undirected Graphs}
\label{sec:almost-undirected}

We consider the generalization of the problem
to (strongly connected) directed graphs. 
The example of a directed cycle shows that one can,
in general, not hope to find a target
using a sublinear number of queries.
Thus, in order to achieve positive results,
additional assumptions must be placed on the graph structure.
Indeed, we saw in Section~\ref{sec:negative} that
deciding, for general strongly connected graphs,
whether a given number of queries is sufficient is 
\Hardness{PSPACE}{complete}.

We show that if the graph is ``almost undirected,''
then the positive result of Theorem~\ref{thm:undirected-weak}
degrades gracefully.
Specifically, we assume that each edge $e$
with weight \Weight{e} is part of a cycle of total
weight at most $\MAXCYCLE \cdot \Weight{e}$.
Notice that for unweighted graphs, this means that each edge $e$
is part of a cycle of at most $\MAXCYCLE$ edges, and specifically
for $\MAXCYCLE = 2$, the graph is undirected.%
\footnote{However, for weighted graphs with $\MAXCYCLE = 2$,
$G$ will not necessarily be undirected.
\Eat{For instance, $G$ could be a cycle of
undirected edges of weight 1 and one directed edge of weight $n - 1$.}}

\begin{theorem}\label{thm:directed-short-cycles}
Algorithm~\ref{alg:undirected-weak-upper} has the following property:
if $G$ is a strongly connected, positively weighted graph, in
which each edge $e$ belongs to a cycle of total weight at most 
$\MAXCYCLE \cdot \Weight{e}$, 
then the algorithm finds the target using at most
\begin{align*}
\frac{1}{\ln \MAXCYCLE - \ln(\MAXCYCLE - 1)} \cdot \ln n
\leq \MAXCYCLE \ln(2) \cdot \log n
\end{align*}
queries in the worst case.
\end{theorem}

\begin{proof}
In Algorithm~\ref{alg:undirected-weak-upper}, the potential is now
defined with respect to directed distances.
To analyze its performance, assume that the algorithm,
in some iteration, queried the node $q$,
and received as a response an edge $e = (q, v)$.
We define the sets $S^+$ and $S^-$ as in the proof
of Theorem~\ref{thm:undirected-weak}. 
As before, $d(v, u) = d(q, u) - \Weight{e}$
for all vertices $u \in S^+$.
On the other hand, there exists a path of total weight at most
$(\MAXCYCLE - 1) \cdot \Weight{e}$ from $v$ to $q$
(since $e$ appears in a cycle of length at most
$\MAXCYCLE \cdot \Weight{e}$).
Thus, for any vertex $u \in S^-$,
$d(v, u) \leq d(q, u) + (\MAXCYCLE - 1) \cdot \Weight{e}$.
Therefore,
\begin{align*}
\Potential{S}{v} \leq \Potential{S}{q} - \Weight{e} \cdot 
\big(\SetCard{S^{+}} - (\MAXCYCLE - 1) \cdot \SetCard{S ^{-}} \big).
\end{align*}
Since $\Potential{S}{q}$ is minimal,
$|S^{+}| \leq (\MAXCYCLE - 1) \cdot |S^{-}|$,
so $|S^{+}| \ \leq  \ \frac{\MAXCYCLE - 1}{\MAXCYCLE} \cdot |S|$.
Thus, after at most 
$\log_{\MAXCYCLE/(\MAXCYCLE - 1)} (n) = 
\frac{\ln n}{\ln \MAXCYCLE - \ln(\MAXCYCLE - 1)}$
queries, the target must be identified.
Finally,
%\[
\begin{align*}
\ln (\MAXCYCLE) - \ln(\MAXCYCLE - 1) 
 =  \int_{\MAXCYCLE-1}^\MAXCYCLE dx/x
 \geq  1/\MAXCYCLE,
\end{align*}
implying that 
$\frac{\ln n}{\ln \MAXCYCLE - \ln(\MAXCYCLE - 1)}
\leq \MAXCYCLE \cdot \ln n =
\MAXCYCLE \ln(2) \cdot \log n$.
\end{proof}

The upper bound of Theorem~\ref{thm:directed-short-cycles} 
is nearly matched,
up to a factor of $O(\log \MAXCYCLE)$,
by the following lower bound.

\begin{proposition}\label{prp:directed-short-cycles-lower-bound}
For any integers $N$ and $\MAXCYCLE \geq 2$,
there exists an unweighted and strongly connected graph $G$ of $n \geq
N$ vertices, 
such that each edge in $G$ belongs to a cycle of length \MAXCYCLE,
and at least $\frac{\MAXCYCLE - 1}{\log \MAXCYCLE} \cdot \log n$ 
queries are required to identify a target in $G$.
\end{proposition}

\begin{proof}
We construct a family of unweighted strongly connected directed graphs
$\NestCycle{\MAXCYCLE}{k}, k = 0, 1, \ldots$ inductively.
\NestCycle{\MAXCYCLE}{0} is a single vertex.
For any $k \geq 1$, \NestCycle{\MAXCYCLE}{k} is obtained from
\MAXCYCLE disjoint copies of \NestCycle{\MAXCYCLE}{k - 1}.
For each of the \MAXCYCLE copies $i = 1, \ldots, \MAXCYCLE$,
let $v_i$ be an arbitrary vertex in that copy.
Now add a directed cycle of length \MAXCYCLE through the vertices $v_i$.
By construction, each edge is part of a cycle of length \MAXCYCLE; 
and by induction, \NestCycle{\MAXCYCLE}{k} is strongly connected
for all $k$. 

We will prove by induction that any strategy requires 
at least $(\MAXCYCLE - 1) \cdot k$ queries
in the worst case to identify the target in \NestCycle{\MAXCYCLE}{k}.
Because $n = \MAXCYCLE ^ k$, this proves a lower bound of 
$\frac{\MAXCYCLE - 1}{\log \MAXCYCLE} \cdot \log n$ on the number of
queries in an $n$-vertex graph.

The base case $k = 0$ of the induction is trivial.
For the induction step, let $G_1, G_2, \ldots, G_{\MAXCYCLE}$ be the
\MAXCYCLE copies of \NestCycle{\MAXCYCLE}{k - 1} that were combined
to form \NestCycle{\MAXCYCLE}{k}. 
For $i < k$, define $v'_i$ to be $v_{i + 1}$ and
let $v'_k$ be $v_1$.
Consider a query of a node $q \in G_i$.  
By construction, $v'_i$
lies on any path from $q$ to any vertex not in $G_i$.
In other words, if an algorithm receives an edge $(q, v)$
lying on a shortest path from $q$ to $v'_i$,
it might learn that the target is not in $G_i$,
but cannot infer which $G_j, j \neq i$ the target is in.

The adversary's strategy is now simple:
for the first $\MAXCYCLE - 1$ queries,
to each query $q \in G_i$, the adversary will give an edge toward $v'_i$,
i.e., the first edge of a shortest path from $q$ to $v'_i$.
At this point, there is at least one $G_i$ such that no vertex in $G_i$
has been queried. The adversary now picks one such $G_i$
(arbitrarily, if there are multiple),
and commits the target to $G_i$.
Then, he continues to answer queries to vertices
in $G_j, j \neq i$ in the same way as before.
Queries to vertices in $G_i$ are answered
with the adversarial strategy for \NestCycle{\MAXCYCLE}{k - 1}.

As a result, after $\MAXCYCLE - 1$ queries, there will always remain at
least one entirely unqueried copy of \NestCycle{\MAXCYCLE}{k - 1}; the
algorithm has no information about the location of the target vertex
in this copy, and by induction hypothesis, it will take at least
$(\MAXCYCLE - 1) \cdot (k - 1)$ queries to find the target.
Adding the $\MAXCYCLE - 1$ queries to reach this point gives us
a total of at least $(\MAXCYCLE - 1) \cdot k$ queries
to find the target in \NestCycle{\MAXCYCLE}{k}.
\end{proof}

%% More Informative Queries on Directed Graphs

\section{More Informative Queries on Directed Graphs} 
\label{sec:more-informative}

Instead of restricting cycle lengths in $G$,
an alternative way to obtain positive results
for directed graphs is to assume that
the responses to queries are more informative.
This approach is motivated by the directed cycle,
where the answer to a query reveals
no additional information to the algorithm. 
If the algorithm could learn not only the outgoing edge
but also the distance to the target,
then a single query would suffice on the cycle.
Indeed, we show that in general, this information is enough
to always find the target using at most $\log n$ queries.
We remark that learning \emph{only} the distance to the target
would not be enough to guarantee even
a sublinear number of queries:
in an unweighted undirected star, when the target is a leaf,
such an answer will reveal no information
except whether the queried node is the target.

Formally, we define more informative responses to queries as follows.
Let $\ReachDist{u}{e}{\ell} =
\SpecSet{v \in \Reach{u}{e}}{d(u, v) = \ell}$.
In response to querying node $q$, the algorithm will be given
an edge $e$ and distance $\ell$ such that 
$t \in \ReachDist{q}{e}{\ell}$.

\begin{theorem}\label{thm:directed-strong-appendix}
Assume that each query reveals the distance from the queried node $q$
to the target $t$ as well as an edge $e = (q, v)$
on a shortest $q$-$t$ path.
Then, there is an efficient algorithm which, for each directed,
strongly connected and positively weighted graph $G$, finds the target
using at most $\log n$ queries.%
\footnote{This upper bound is tight even for complete binary trees.}
\end{theorem}

\begin{proof}
The algorithm is nearly identical to
Algorithm~\ref{alg:undirected-weak-upper}.
It also queries a vertex $q$ minimizing a potential function
of the set $S$ of remaining candidate vertices at each iteration.
However, the potential function takes a different form.

For each distance $\ell$, let $\AtDistance{\ell}{v}$ be the set of
vertices at distance at most $\ell$ from $v$.
Define $\DirPotential{S}{v}$ as the minimum distance $\ell$
such that strictly more than half of the vertices of $S$
are in $\AtDistance{\ell}{v}$.
When receiving an answer of $(e,\ell)$ in response to a vertex query $q$,
the algorithm updates $S$ to $S \cap \ReachDist{q}{e}{\ell}$.
Except for using $\DirPotential{S}{v}$ in place of
$\Potential{S}{v}$ in Line~\ref{line:potential}, receiving 
the distance $\ell$ in addition to the edge $e$ in
Line~\ref{line:response},
and updating $S \AlgAssign S \cap \ReachDist{q}{e}{\ell}$ 
in Line~\ref{line:undirected-new-S}, the algorithm is identical to
Algorithm~\ref{alg:undirected-weak-upper}.

Similar to the analysis of Algorithm~\ref{alg:undirected-weak-upper},
we show that the size of $S$ decreases
by at least a factor of 2 in each iteration.
Consider one iteration in which a vertex $q$ is queried,
and the response is an edge $e = (q, v)$ and distance $\ell$.
We partition $S$ into the remaining candidates 
$S^{+} = S \cap \ReachDist{q}{e}{\ell}$,
and the eliminated vertices $S^{-} = S \setminus S^{+}$.

Assume for contradiction that $|S^{+}| > |S|/2$.
Because more than half of the vertices of $S$ are thus at
distance \emph{exactly} $\ell$ from $q$, we obtain that
$\DirPotential{S}{q} = \ell$.
And because $v$ lies on a shortest $q$-$u$ path for any $u \in S^{+}$,
the distance from $v$ to each $u \in S^+$ is strictly less
than $\ell$.
Thus, more than $\Half{|S|}$ vertices of $S$
are at distance less than $\ell$ from $v$,
implying that
$\DirPotential{S}{v} < \ell$. But this contradicts the choice of $q$
as minimizing \DirPotential{S}{v}.
\end{proof}

%% Edge Queries on Trees

\section{Edge Queries on Trees} \label{sec:edge-queries}

As mentioned in the introduction,
a version of the problem that has been studied more frequently 
\cite{linial-saks:1985:searching,%
lam-yue:2001:edge-ranking-linear,%
BenAsher-farchi:1997:trees,%
BenAsher-farchi-newman:1999:trees-edge-poly,%
onak-parys:2006:trees-vertex-linear,%
mozes-onak-weimann:2008:trees-edge-linear},
is one in which the algorithm is allowed to query
the \emph{edges} of a graph (most often: a tree).
For the case of trees, by querying an edge $e = (u, v)$,
the algorithm learns whether the target is in the subtree
rooted at $u$ or the one rooted at $v$.

Ben Asher and Farchi~\cite{BenAsher-farchi:1997:trees} proved,
for a tree with maximum degree \MAXDEGREE,
an upper bound of 
$\log_{\MAXDEGREE/(\MAXDEGREE-1)} (n) = \Theta(\MAXDEGREE \log n)$
and a lower bound of $\frac{\MAXDEGREE-1}{\log \MAXDEGREE} \log n$
on the number of queries required,
which --- contrary to the authors' claim ---
leaves a small gap of $\Theta(\log \MAXDEGREE)$.
In this section,
we prove that the lower bound is practically tight:

\begin{theorem}\label{thm:tree-edge-upper}
For any tree $T$ with maximum degree $\MAXDEGREE > 1$,
there is an adaptive algorithm for finding the target using at most
$1 + \frac{\MAXDEGREE-1}{\log (\MAXDEGREE+1) - 1} \log n$
edge queries.
\end{theorem}

\begin{proof}
The algorithm is quite similar to the algorithm for vertex queries
in Section~\ref{sec:deterministic},
in that it repeatedly queries a median node $v$. 
However, since a vertex in the given tree
cannot be queried directly, the algorithm
instead finds a median node $v$ and
queries the edge $(v, u)$,
where $u$ is the neighbor of $v$
maximizing the size of its rooted subtree.

More precisely, the algorithm is as follows:
Given a tree $T$, the algorithm initializes $T' = T$.
In each iteration, it considers a separator node $v$
(a vertex such that each subtree (in fact, connected component)
of $T' \setminus \Set{v}$ contains at
most $|T'|/2$ vertices),
breaking ties in favor of the separator node from the previous iteration.
That is, if the last chosen separator is still a valid separator
for the current tree $T'$, the algorithm uses it for the next iteration;
otherwise, it picks a new separator node arbitrarily.
Among all neighbors of $v$, let $u$ be one maximizing the
size of the subtree $T_u$ of $T'$ rooted at $u$. The algorithm then
queries the edge $(v, u)$,
and depending on whether the target is in
the subtree rooted at $v$ or at $u$, updates $T'$ to either 
$T' \setminus T_u$ or $T_u$. When only a single vertex remains,
the algorithm returns it.
More formally, the algorithm is given as
Algorithm~\ref{alg:tree-edge}.

\InsertAlgorithm{Target Search using Edge Queries $(T)$}{alg:tree-edge}{
%\STATE {\textbf{Input:} tree $T$}
\STATE{$T' \AlgAssign T$.}
\STATE{$v$ \AlgAssign a separator node of $T'$: each subtree
of $T' \setminus \Set{v}$ contains at most $|T'|/2$ vertices.}
\WHILE{$|T'| > 1$}
	\IF{$v$ is not in $T'$ any more or it is not a separator for $T'$.}
		\STATE {$v$ \AlgAssign a separator for $T'$}.
	\ENDIF
	\STATE{Let $u_1, u_2, \ldots$ be the neighbors of $v$ in $T'$,
	and $T_1, T_2, \ldots$ the subtrees of $T'$
	rooted at $u_1, u_2, \ldots$.}
	\STATE{$i \AlgAssign \argmax\limits_{j} \Set{|T_j|}$.}
	\STATE{Query the edge $(v, u_i)$.}
	\IF{the target is in $T_i$}
		\STATE $T' \AlgAssign T_i$.
	\ELSE
		\STATE $T' \AlgAssign T \setminus T_i$.
	\ENDIF
\ENDWHILE
\RETURN{the unique vertex in $T'$}.
}

First off, recall that the existence of a separator node
was proved by Jordan~\cite{jordan:1869:assemblages}.
For the purpose of analysis, we divide the execution of the algorithm
into phases: maximal intervals during which the same vertex $v$
is used as a separator.
Notice that there may be multiple phases during which the same
vertex is chosen --- however, no two such phases can be adjacent
by maximality.

Consider one phase and its corresponding vertex $v$;
suppose that $v$ was used $k$ times,
and had degree $d$ in $T'$ at the beginning of the phase.
Let $u_1, u_2, \ldots, u_d$ be the neighbors of $v$ in $T'$,
with subtrees $T_1, T_2, \ldots, T_d$.
Without loss of generality, assume that
$|T_1| \geq |T_2| \geq \ldots \geq |T_d|$,
so that during those $k$ queries, the edges
$(v, u_1), (v, u_2), \ldots, (v, u_k)$ were queried.
If the phase is the last phase, and at some point, the tree only
contains two vertices (and one edge), we treat the last step,
determining the target with one more query, separately.
We now distinguish three cases, based on how the phase ended:

\begin{enumerate}
\item 
The \Ordinal{k} query revealed that the target is in $T_k$,
so that $v$ was permanently removed from $T'$
(along with all $T_i$ for $i \neq k$).
By the ordering of sizes, we have $|T_i| \geq |T_k|$ for all $i < k$.
Furthermore, because $v$ is a separator
at the \Ordinal{k} query of the phase,
$1 + \sum_{i = k + 1}^d |T_i| \geq |T_k|$.
Thus, $|T'| \geq (k + 1) |T_k|$
where $|T'|$ is the size of the tree at the beginning of the phase;
in other words, using $k$ queries,
the size of the remaining tree was reduced by a factor at least $k + 1$.

\item 
The \Ordinal{k} query did not place the target in $T_k$,
but the removal of $T_1, \ldots, T_k$ resulted in $v$ not being a
separator any more.
By the ordering of sizes, we have $|T_i| \geq |T_{k+1}|$
for all $i \leq k$.
Let $r = 1 + \sum_{i \geq k+1} |T_i|$ be the size of the
remaining tree after the removal of $T_1, \ldots, T_k$.
Because $v$ is not a separator any more, $r < 2 |T_{k+1}|$,
while $|T'| \geq k |T_{k+1}| + r$,
where $|T'|$ is the size of the tree at the beginning of the phase.
Thus, using $k$ queries, the size of the remaining tree was reduced
by a factor at least $\frac{k |T_{k+1}| + r}{r} \geq \frac{k+2}{2}$.

\item 
$T'$ consisted of a single node (the algorithm is done), or two
nodes with one edge (the algorithm is done with one more query).
\end{enumerate}

Thus, in the first two cases,
using $k$ queries, the size of the remaining tree
decreases by a factor at least $\frac{k+2}{2}$.
Now, for each phase $j$, let $k_j$ be the length of the phase.
The total number of queries is $\sum_j k_j$, and the constraint
imposed by the size decrease is that
$\prod_j \frac{k_j + 2}{2} \leq n$.
By writing $\frac{k_j+2}{2} = e^{x_j}$, we obtain a convex
objective function $\sum_j (2e^{x_j}-2)$
subject to an upper bound
$\sum_j x_j \leq \ln n$, which is
maximized by making each $x_j$ (and thus $k_j$)
as large as possible or equal to 0.
In other words, we obtain an upper bound by setting 
$k_j = \MAXDEGREE - 1$ for all $j$.
Notice that the only case in which the algorithm queries all
$d$ edges incident on a node is when there are exactly two nodes
remaining before the final query; this case is treated separately.

The number of phases is then
$\log_{(\MAXDEGREE+1)/2} n =
\frac{\log n}{\log (\MAXDEGREE+1) - 1}$.
Thus, the total number of queries, counting the possible final
query, is at most 
$1 + \frac{\MAXDEGREE-1}{\log(\MAXDEGREE+1)-1} \cdot \log n$,
completing the proof.
\end{proof}

%%%%%%%%%%%%%%%%%%%%%%%%%%%%%%%%%%%%%%%%%%%%%%%%%%

\end{document}